\documentclass[journal]{IEEEtran}
\usepackage{verbatim}
\usepackage{url}
\usepackage{graphicx}
\usepackage{balance}
\usepackage{comment}
\usepackage{makecell}
\usepackage{subfigure}
\usepackage{amsmath}
\usepackage{algorithmic}
\usepackage{algorithm}
\usepackage{array}
\usepackage{amsthm}
\usepackage{wrapfig}
\usepackage[suffix=]{epstopdf}
\usepackage{epsfig}

\newtheorem{theorem}{Theorem}[section]

\newtheorem{proposition}[theorem]{Proposition}

\ifCLASSINFOpdf
\else
\fi

\begin{document}
%
\title{Free Side-channel Cross-technology Communication in Wireless Networks}
%
%
%

\author{Song~Min~Kim,
        Shigemi~Ishida,
        Shuai Wang,
        and~Tian~He$^*$
\thanks{This work was supported in part by the National Science Foundation under grants CNS-1525235 and CNS-1444021. A conference paper~\cite{songminFREEBEE} containing preliminary results of this paper appeared in ACM MobiCom 2015.}
\thanks{$^*$ Corresponding author.}
\thanks{S. M. Kim is with the Department
of Computer Science, George Mason University, Fairfax,
VA, 22030 USA e-mail: song@gmu.edu.}
\thanks{S. Ishida was a visiting scholar at the Department
of Computer Science and Engineering, University of Minnesota, Minneapolis,
MN, 55455 USA e-mail: ishida@f.ait.kyushu-u.ac.jp.}
\thanks{S. Wang and T. He are with the Department
of Computer Science and Engineering, University of Minnesota, Minneapolis,
MN, 55455 USA e-mail: \{shuaiw, tianhe\}@cs.umn.edu.}
}

%
%

\markboth{IEEE/ACM TRANSACTIONS ON NETWORKING}%
{Shell \MakeLowercase{\textit{et al.}}: Bare Demo of IEEEtran.cls for IEEE Journals}
%



\maketitle

\begin{abstract}
Enabling direct communication between wireless technologies immediately brings significant benefits including, but not limited to, cross-technology interference mitigation and context-aware smart operation. To explore the opportunities, we propose FreeBee -- a novel cross-technology communication technique for direct unicast as well as cross-technology/channel broadcast among three popular technologies of WiFi, ZigBee, and Bluetooth. The key concept of FreeBee is to modulate symbol messages by shifting the timings of periodic beacon frames already mandatory for diverse wireless standards. This keeps our design generically applicable across technologies and avoids additional bandwidth consumption (i.e., does not incur extra traffic), allowing continuous broadcast to safely reach mobile and/or duty-cycled devices. A new \emph{interval multiplexing} technique is proposed to enable concurrent bro\-adcasts from multiple senders or boost the transmission rate of a single sender. Theoretical and experimental exploration reveals that FreeBee offers a reliable symbol delivery under a second and supports mobility of 30mph and low duty-cycle operations of under 5\%.

\end{abstract}


%
\IEEEpeerreviewmaketitle

\section{Introduction}\label{sec:introduction}

\IEEEPARstart{D}{iverse} wireless devices used today are specialized to enrich different domains of our daily lives. However, wireless technologies are victims of their own success: spectrum sharing among incompatible  technologies has led to a severe \emph{wireless coexistence problem}~\cite{GollakotaAKS11, RadunovicCG12, srinivasan2010empirical, zhang2011enabling, zhang2013cooperative, ZhaoWZZL14}, significantly impacting the networking performance.

We begin with the recognition that this coexistence is indeed double-sided: although it may cause inefficiency and unfairness in spectrum utilization, it also provides new opportunities because the standards for individual technologies are specialized and hence possess strengths in different areas that are, often the weaknesses of the others. For example, while WiFi has access to a virtually unlimited amount of information via the Internet, it consumes a considerable amount of power, causing battery problems in mobile devices~\cite{Balasubramanian09, ding2013characterizing}. Conversely, the ZigBee network often operates as a stand-alone and has limited information, but is extremely energy efficient. Thus, both networks can be enhanced via mutual supplementation, demonstrating the positive side of coexistence.

In this paper, we propose FreeBee, a cross-technology communication framework that is generic and transparent (i.e., no extra traffic). Our design aims to mitigate the detrimental effect of coexistence while exploring the opportunities behind it. As such, FreeBee sheds the light on the opportunities that cross-technology communication has to offer including, but not limited to, cross-technology interference mitigation and context-aware smart operation. Specifically we achieve this via embedding symbol into beacons by shifting their transmission timing. Although the concept of modulating via signal timings is known as PPM (Pulse Position Modulation), legacy PPM supports only communication between homogeneous devices and requires precise pulse timing, which can be hardly satisfied in wireless coexistence environments with mainly contention-based MACs.

Existing cross-technology communication works~\cite{ChebroluD09, ZhangS13} are technology-specific and require dedicated packets for communication, burdening already-crowded channels with further overhead. In contrast, Freebee utilizes mandatory beacons widely adopted among wireless technologies~\cite{iBeacon,IEEE80211,IEEE802154}, achieving a generic and free-side-channel design. In summary, our original contribution is three-fold:

\begin{itemize}
\item We propose FreeBee,  a novel cross-technology communication framework that allows direct communication between heterogeneous senders and receivers.
In addition, FreeBee allows heterogeneous devices to receive broadcast simultaneously from a sender with overlapping frequencies (e.g., Bluetooth to WiFi and ZigBee) and support a sender with a wider bandwidth
(e.g., WiFi) to reach multiple narrower-band receivers (e.g., WiFi to multi-channel ZigBee).

\vspace{0.1in}

\item FreeBee requires no hardware modification and does not introduce dedicated traffic.  Its existence is transparent to legacy wireless systems.  Our new\emph{ interval multiplexing scheme} supports concurrent transmission and reception of multiple signals.

\vspace{0.1in}

\item We present three prototype implementations: WiFi, ZigBee, and Bluetooth. Results suggest that FreeBee offers reliable symbol delivery within less than a second and supports mobility up to 30mph and duty cycle operation of under 5\%. We also demonstrate a practical use of FreeBee: Inspecting real WiFi deployment patterns in a shopping mall area, FreeBee was found to save 78.9\% of the energy otherwise wasted by the WiFi interface.
\end{itemize}


The paper is organized as follows. Section~\ref{sec:motivation} presents motivations. 
The FreeBee design and features are introduced in Sections~\ref{sec:designbasics} and~\ref{sec:DiscussionChallenges}. 
Sections~\ref{sec:analysis} and~\ref{sec:evaluation} evaluate performance analytically and empirically.
We summarize related work in Section~\ref{sec:relatedworks} and conclude in Section~\ref{sec:conclusion}.

\section{Motivation} \label{sec:motivation}

This section demonstrates a few example use cases among a wide range of benefits the FreeBee technology has to offer.

\begin{figure}[h]
  \centering
  \includegraphics[width=0.47\textwidth]{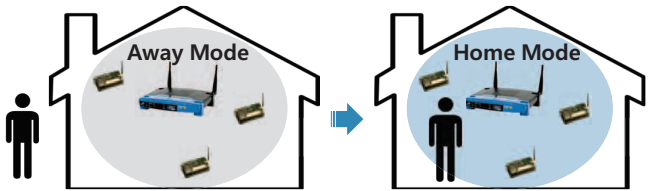}
  \vspace{-2mm}
  \caption{Context-aware home automation}
  \label{fig:context_aware}
\end{figure}

\noindent$\bullet$ \textbf{Benefit to ZigBee -- Smart homes}:
FreeBee enables information sharing directly between technologies, without assistance from dual-radio gateways which are costly ($>$150 USD) and mostly unavailable in real-life settings. One such type of information is user presence, which is accessible by WiFi AP by observing the nodes associated to it. Sharing this information enables other networks to provide context-aware service.
Figure~\ref{fig:context_aware} demonstrates the example of a smart home with ZigBee-assisted appliances. Home WiFi AP first determines whether the resident is away or home (i.e., his/her smart phone is associated or not). Using FreeBee, this information is broadcast from the AP to all the ZigBee nodes inside the home to drive them to the appropriate operation mode. For instance, once the resident leaves they turn to ``away mode'', such as lowering home temperature to an energy-economic value.

\begin{figure}[h]
  \centering
  \includegraphics[width=0.35\textwidth]{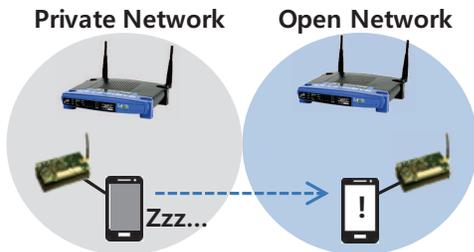}
  \vspace{-2mm}
  \caption{Wake on selective WLAN}
  \label{fig:selective_wlan}
\end{figure}

\noindent$\bullet$ \textbf{Benefit to WiFi -- Mobile devices}: Operating a WiFi Network Interface Card (NIC) continuously depletes precious energy in portable devices~\cite{Balasubramanian09, ding2013characterizing}. To tackle this, ZiFi~\cite{ZhouXXSM10} suggests attaching a low-power ZigBee radio to the device to wake up the WiFi NIC whenever it detects the existence of \emph{any} WiFi AP. While this approach can significantly reduce standby energy, we believe that further savings can be achieved.

Nowadays most of metropolitan areas are overloaded with WiFi APs, many of which are private. Thus, it is still a waste of energy to blindly wake up a WiFi NIC when an arbitrary AP shows up, without knowing if it is accessible or not. As shown in Figure~\ref{fig:selective_wlan}, we avoid this issue by embedding 1-bit accessibility information (i.e., open/private) into WiFi beacons and allowing the attached low-power radio (Bluetooth or ZigBee) to capture this information through FreeBee. Accordingly, the WiFi NIC can wake up only when it finds an open AP. We refer to this function as \emph{wake on selective WLAN}. We note that this approach can be extended to limiting WiFi wakeup only to discovering an AP that fits the user's interest; for instance, when there is an AP that matches a user-defined SSID.

\begin{figure}[h]
  \centering
  \includegraphics[width=0.47\textwidth]{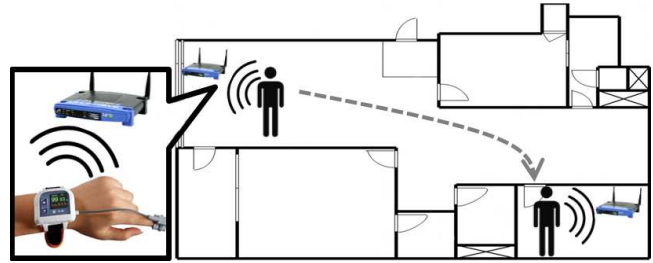}
  \vspace{-2mm}
  \caption{Real-time patient monitoring}
  \label{fig:LBS}
\end{figure}

\noindent$\bullet$ \textbf{Benefit to Bluetooth -- Health care}: Taking advantage of its low-power operation, Bluetooth technology is widely used in portable medical devices, including glucose and heart monitors~\cite{carroll2007continua}. Although FreeBee is not designed to transfer a large volume of medical data, it enables health alerts by embedding urgent information into Bluetooth beacons. The ubiquitous WiFi coverage in most indoor environments today provides a continuous alert service even if patients are away from their Bluetooth-enabled medical station, as shown in Figure~\ref{fig:LBS}. This figure also shows that FreeBee can offer the location of the patient via geolocation provided by WiFi AP, allowing accurate and timely medical actions in case of an emergency.

\begin{figure}[h]
  \centering
  \includegraphics[width=0.45\textwidth]{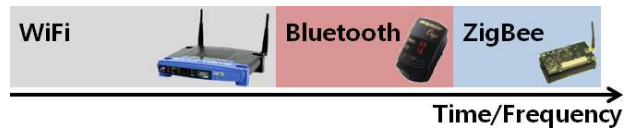} \vspace{-2mm}
  \caption{Cross-technology coordination}
  \label{fig:ch_coordination}
\end{figure}

\noindent$\bullet$ \textbf{Benefit to all -- Channel efficiency}: All WiFi, ZigBee, and Bluetooth networks can benefit from FreeBee via cross-technology channel coordination. As a result of evolving separately without considering each other, channel access schemes in heterogeneous wireless technologies are incompatible, leading to a severe CTI~\cite{HsuWK15, LiangPLT10, ZhangS13, ZhangL13a}. This cross-technology channel access problem can be addressed by allowing explicit communication among different technologies. As demonstrated in Figure~\ref{fig:ch_coordination}, FreeBee essentially allows TDMA or FDMA among heterogeneous wireless platforms, alleviating CTI. For instance, FreeBee realizes mechanisms similar to NAV (Network allocation vector) or RTS/CTS in WiFi for spectrum allocation that is global across technologies.



\section{Framework Design} \label{sec:designbasics}

An overview of FreeBee is presented, followed by technical details including modulation and demodulation.
Without loss of generality, we use communication between WiFi (sender) and ZigBee (receiver) to illustrate the generic design of FreeBee.

\subsection{Design Overview} \label{sec:design_overview}
We propose a cross-technology communication framework in which symbols are embedded into the \emph{timing} of beacon frames. Specifically, we slightly shift the transmission time (advance or delay) of beacon frames, a configurable setting in most WiFi APs deployed today, simply via HTTP protocol. The shift is made in the units of 1.024ms, compliant to the 802.11 standard unit used in beacon scheduling, known as TBTT (Target Beacon Transmission Time)~\cite{IEEE80211}. To ensure a free side-channel operation, FreeBee shifts timing in such a way that the average interval remains the same  as the original setting. Thus, the proposed communication framework is not only transparent, but also does not consume additional bandwidth or energy. This unique aspect enables important information to be broadcast continuously, safely reaching mobile and/or duty-cycled ZigBee receivers whose presence or active periods are a priori unknown to the sender.

However, WiFi beacons cannot be directly captured and recognized by ZigBee nodes, due to incompatible PHY layers. Instead, as WiFi coexists with ZigBee on the 2.4GHz ISM band, a ZigBee receiver statically detects the position of the WiFi beacons in the wireless channel, using its RSSI sensing capability. We note that, as a foundational function for MAC techniques including CSMA, RSSI sampling is common among varieties of wireless standards (e.g., Bluetooth).

\begin{figure}[h]
  \centering
  \includegraphics[width=0.48\textwidth]{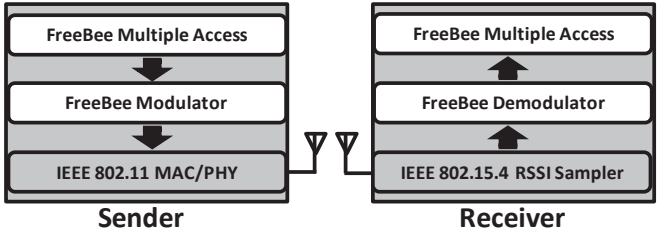}
  \caption{Architecture and scope of FreeBee depicted in white boxes}
  \label{fig:FreeBee_Diagram}
  \vspace{-2mm}
\end{figure}

Figure~\ref{fig:FreeBee_Diagram} depicts the overall architecture and the scope of this paper. Our design spans modulation/demodulation techniques, as well as a multiple access scheme for concurrent communications. We note that the WiFi-ZigBee FreeBee design is based on 802.11 and 802.15.4 standards, and is compatible with 802.11 variants (i.e., a/b/g/n) and 802.15.4-compliant nodes (e.g., TelosB and MICAz) while requiring no hardware modification. In fact, FreeBee can be adopted to enable communication between any heterogeneous wireless platforms as long as the channel is shared. In the following, we first propose a basic version of FreeBee with the simplest form, followed by elaborated designs that enhance the basic version.

\subsection{Basic FreeBee} \label{sec:basic_freebee}
This section describes the basic version of modulation/demodulation. This design assumes that the unmodulated position of the beacon is found during the initial network setup, which we refer to as the \emph{reference position}. This position can be simply obtained by running FreeBee when the AP is sending beacons in their original timing. Modulated beacon positions found later are compared to the reference position for symbol interpretation (i.e., demodulation).

\begin{figure}[h]
\vspace{-2mm}
  \centering
  \includegraphics[width=0.45\textwidth]{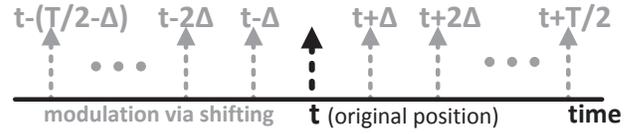}
  \caption{The symbol is embedded by shifting the beacon from its reference position, $t$, where the degree of shift ranges between ($-T/2$, $T/2$]. }
  \label{fig:Synchronized}
  \vspace{-4mm}
\end{figure}

\subsubsection{Modulation} \label{sec:modulation}
For timely advertisement, the 802.11 standard requires APs to periodically broadcast beacons. FreeBee establishes a free side-channel by embedding symbols within the \emph{transmission timing} of these mandatory packets. Referring to Figure~\ref{fig:Synchronized}, let's consider a beacon whose reference position is at $t$, where the interval is $T$. Applying FreeBee, we \emph{shift} the beacon from its reference position in the range of ($-T/2$, $T/2$] to indicate the symbol to be delivered.
The amount of information that can be embedded is determined by $T$ and the granularity of shift, indicated by $\Delta$ in Figure~\ref{fig:Synchronized}. We set $\Delta$ as $1.024ms$ following the beacon scheduling granularity as defined by the 802.11 standard (We note that the information amount can easily be increased by adopting smaller $\Delta$). Under this setting the typical $T$ of $102.4ms$, adopted in the majority of legacy WiFi APs, corresponds to 100$\Delta$, indicating that the beacon can be positioned at 100 different time instances. Thus, beacon shift can express $\left\lfloor log_{2}100\right\rfloor=6$ bits.



Due to incompatible PHY layers, the ZigBee receiver is unable to decode the beacons and thus cannot detect the presence of beacons directly. Therefore we statistically locate beacons by their periodic repetition. For instance, to deliver a FreeBee symbol corresponding to $t-\Delta$ in Figure~\ref{fig:Synchronized}, multiple consecutive beacons are shifted for the same amount (i.e., beacons are transmitted at $t+T-\Delta$, $t+2T-\Delta$, and so on). The required number of beacon repetitions per symbol is decided by the channel noise, which is analyzed in detail in later sections. Lastly, we note that the beacon interval still remains at $T$ in the process of FreeBee transmission, indicating free side-channel operation.

\subsubsection{Demodulation} \label{sec:demodulation}
Here we describe how FreeBee captures and interprets position-modulated beacons for successful demodulation, especially under channel noise (i.e., other ongoing traffic within the same band).
FreeBee demodulation starts from sampling the energy in the channel. This is done by consecutively recording the values obtained from the RSSI register on an 802.15.4-compliant RF chip on the ZigBee node. Upon recording a stream of RSSIs, the captured values are quantized to binaries--0 if below threshold and 1 if above--to indicate clear and busy channels. The threshold is set to be -75dBm following the CCA (Clear Channel Assessment) threshold for the 802.15.4 standard~\cite{IEEE802154}. We note that WiFi also runs CCA, where the threshold is -82dBm~\cite{IEEE80211}. For simplicity, we hereafter will refer to the binary value simply as RSSI. Furthermore, as a RSSI sample in ZigBee is a measurement spanning for $128us$, the sampling rate is set to be $1/128us=7.8KHz$ to avoid time gaps between samples while keeping the rate at its minimum.

\begin{figure}[h]
\vspace{-2mm}
  \centering
  \includegraphics[width=0.43\textwidth]{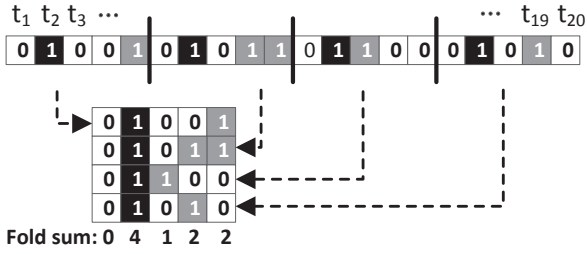}   \vspace{-2mm}
  \caption{Folding example: series of RSSI samples expressed as boxes. Both black and gray indicate a busy channel where the former is a periodic beacon signal with $\lambda=5$, while the latter is random noise induced by traffic or interference. White represents an idle channel. By \emph{folding} the series into a matrix with $P=5$ $(=\lambda)$, the black boxes align column-wise.}
  \label{fig:Folding}
  \vspace{-2mm}
\end{figure}

We then apply \emph{folding} to the obtained RSSI vector, a signal processing technique that allows detecting periodic signal under noise. We note that this technique was originally introduced in~\cite{Staelin69} and was recently featured by ZiFi~\cite{ZhouXXSM10} to detect the presence of WiFi AP.
Given a sampled RSSI vector, folding by $P$ simply cuts the vector into sub-vectors of equal lengths of $P$ and stacks them to yield a matrix. An example of folding is shown in Figure~\ref{fig:Folding}, where a sampled RSSI vector of length 20 is considered. Let the interval of the beacons captured in the vector be $T$ seconds, and the number of samples in $T$ as $\lambda$
In the example $\lambda=5$, and upon folding by $P=\lambda$, RSSIs of beacons (in black) align in a column. the column-wise sum is referred to as the \emph{fold sum}, where the column with the highest fold sum indicates the position of the beacon. Note that the fold sums are likely to be smaller in other columns, as they are induced by either random (thus aperiodic) traffic or beacons with different intervals.


\begin{figure}[h]
\vspace{+2mm}
\centering
\subfigure[\small{Reference Position}]{\label{fig:Filter_Unmodulated}
\includegraphics[width=0.21\textwidth, height=0.15\textheight]{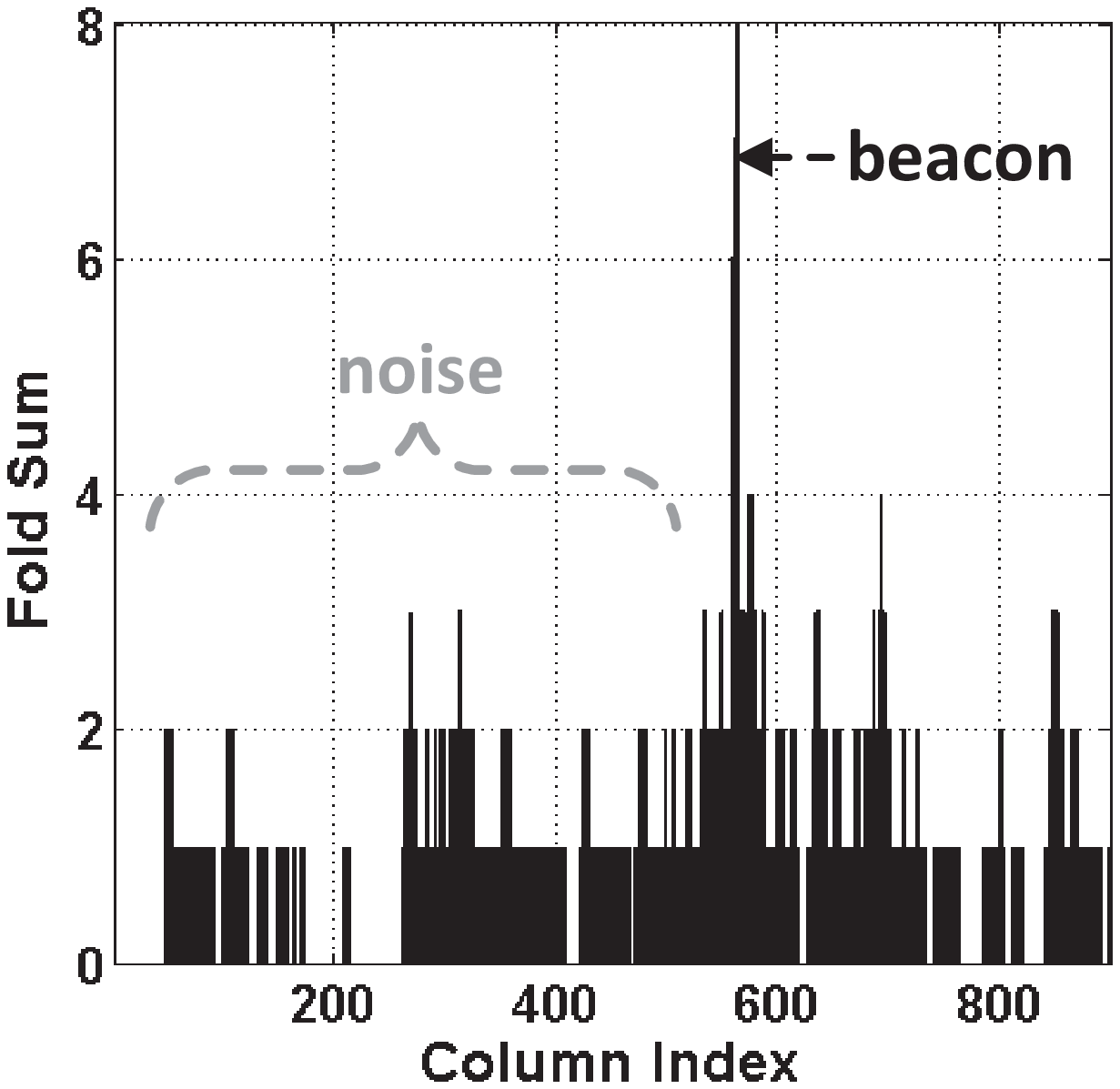}}
\subfigure[\small{Modulated Position}]{\label{fig:Filter_Modulated}
\hspace{3mm}
\includegraphics[width=0.21\textwidth, height=0.15\textheight]{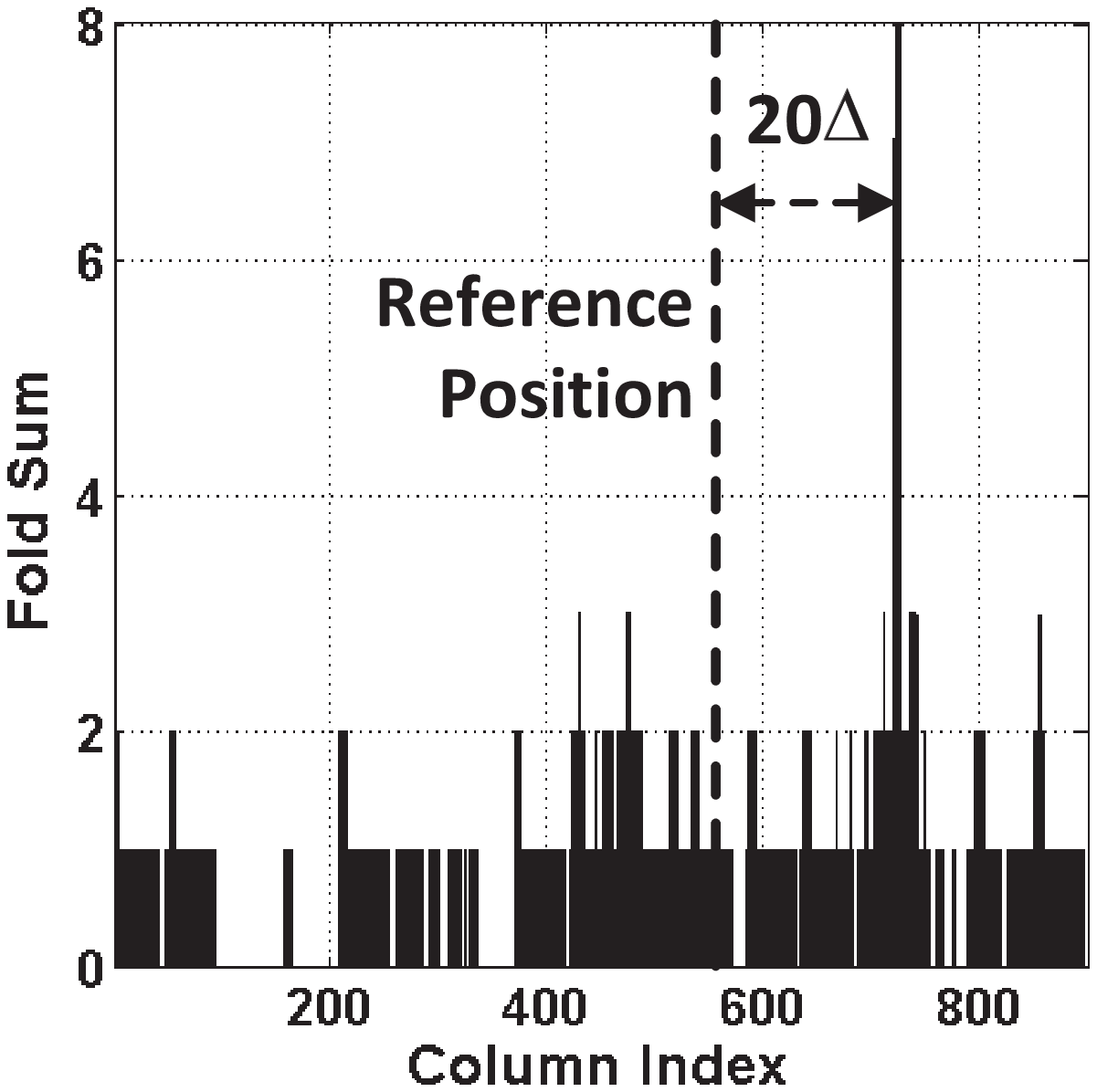}}
\vspace{-2mm}
\caption{Example of FreeBee demodulation in practice, when $T=113\Delta =115.7ms$ and $\lambda=904$}
\label{fig:FilterResult}
\end{figure}

Figure~\ref{fig:FilterResult} presents an example of demodulating FreeBee symbol (20$\Delta$) in practice. To sum up, FreeBee demodulation is process of finding the column corresponding to beacon position, which can easily be achieved by folding and simply picking the column with the maximum fold sum. This same process is used to learn the reference position of the beacon during network initialization and to find the modulated position. The difference of the two positions indicates the symbol within. Other harmonic analysis techniques, such as FFT (Fast Fourier Transform) and autocorrelation, do not yield the position (i.e., phase) of the beacon, and hence are not suitable for our purpose.

\begin{figure}[h]
\vspace{-2mm}
  \centering
  \includegraphics[width=0.4\textwidth]{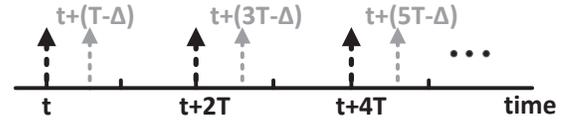}
  \caption{A-FreeBee: the positions of \emph{every other} beacons are shifted (in gray), whereas the degree of shift is $\Delta$ in the example.}
  \label{fig:Asynchronous}
\vspace{-2mm}
\end{figure}

\subsection{Enhanced Feature \#1: Asynchrony} \label{sec:async_freebee}

The basic FreeBee design embeds the symbol as a beacon time shift from the reference position, a concept that requires learning the position beforehand.
We relax this condition to introduce A-FreeBee (Asynchronous FreeBee), freeing our design from any prior knowledge to offer instantaneous communication.


Figure~\ref{fig:Asynchronous} shows Modulation of A-FreeBee. Applying A-FreeBee to beacons with an interval of $T$, beacons are scheduled to construct two streams of beacons (black and gray) with the same interval of $2T$, where one stream (i.e., gray beacons) is a shifted version of the other (i.e., black beacons). This is achieved by shifting \emph{every other} beacon by the amount that corresponds to the symbol to be delivered. The figure demonstrates a case where the symbol corresponds to $\Delta$, indicating one unit of shift. We note that A-FreeBee is also a free-side-channel, as the average interval between consecutive beacons is still $T$.


Under A-FreeBee design, the reference position is no longer required; it simply looks for two beacon streams with the same period by folding with $P=2\lambda$. The embedded symbol is interpreted directly from the phase difference, i.e.,  $Two$ columns with the first and second highest fold sums are found, where the distance between them indicates A-FreeBee symbol.
\begin{figure}[h]
  \centering
  \includegraphics[width=0.33\textwidth]{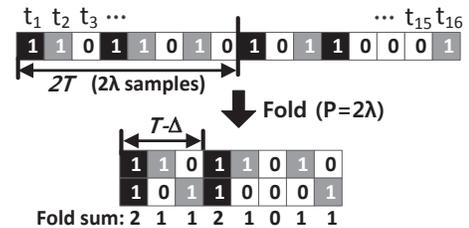}
  \caption{Folding example for A-FreeBee: two high fold sum columns are detected via folding.}
  \label{fig:Folding_Asynchronous}
\end{figure}
A demodulation example of A-FreeBee is shown in Figure~\ref{fig:Folding_Asynchronous} for a RSSI vector of length 16. Two beacon streams are depicted in black boxes, and the gray boxes represent noise. 
Two columns with high fold sums are found by folding with $P=2\lambda=8$, where the distance between the two columns is $T-\Delta$. Noting that the distance would simply be $T$ before modulation, the amount of shift (i.e., the symbol) is therefore $\Delta$. Figure~\ref{fig:Real_life}  below demonstrates A-FreeBee demodulation in practice, where the conveyed symbol is $20\Delta$.

\begin{figure}[h]
  \centering
  \includegraphics[width=0.35\textwidth]{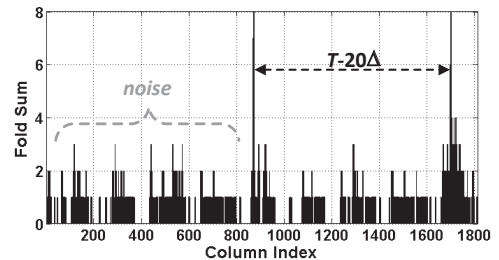}\vspace{-2mm}
  \caption{A-FreeBee demodulation ($T=113\Delta$)}
  \label{fig:Real_life}
  \vspace{-4mm}
\end{figure}

\pagebreak

The asynchronous version features several advantages over the basic FreeBee, including that it (i) requires no synchronization, (ii) is robust to clock drifts, and (iii) supports instantaneous communication without prior knowledge. We note that all these improvements come with a trade-off in data rate; A-FreeBee, compared to the basic FreeBee, requires collecting a higher number of beacons to form \emph{two} high fold sums instead of one. The relationship between the performances of FreeBee and A-FreeBee are analyzed both theoretically and empirically in later parts of the paper.

\subsection{Enhanced Feature \#2: Concurrency} \label{sec:freebee_mac}

Under multiple (A-)FreeBee senders, selecting the same or arbitrary intervals may lead the signals to tangle and cause demodulation errors. In this section, we address this issue to allow \emph{simultaneous transmissions} of an arbitrary number of (A-)FreeBee symbols such that each of them can be safely demodulated at the receiver.


\begin{figure}[h]
\vspace{-2mm}
  \centering
  \includegraphics[width=0.5\textwidth]{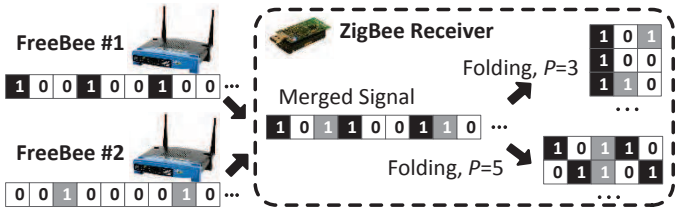}
  \caption{Demodulating interval-multiplexed FreeBee symbols.}
  \label{fig:Multiplexing}
  \vspace{-4mm}
\end{figure}

\subsubsection{Interval Multiplexing}
Recall that, according to the 802.11 standard, beacon intervals are defined in the unit of $\Delta$(=1.024$ms$). Let beacon intervals of $n$ neighboring APs be $x_1\Delta$, $x_2\Delta$, ..., $x_n\Delta$. Then, FreeBee allows simultaneous communication of $n$ APs if {$x_1$, $x_2$, ... $x_n$} are pair-wise co-primes.
We refer to this as\emph{ interval multiplexing}.


Figure~\ref{fig:Multiplexing} demonstrates a scenario of  interval multiplexing/demultiplexing,  where  two FreeBee senders with intervals $T_1=3\Delta$ (in black) and $T_2=5\Delta$ (in gray) introduce a vector of merged RSSI signals at a receiver, and this receiver utilizes interval de-multiplexing to demodulate.  Specifically, by folding with $P=3$ and 5 and looking for the column with the highest fold sum, the receiver can detect the position of the beacons as if the other signal does not exist. This is because 3 and 5 are co-primes, and no longer holds when they are not; for example, consider a sampled RSSI vector including beacons with intervals $T_1=2\Delta$ and $T_2=4\Delta$. When folded by $P=4$, both beacons will form high fold sums, causing demodulation error. We note that while the figure shows only two senders for clarity, this idea can be extended to $n$ senders as long as the intervals are pairwise co-prime. The rationale behind this scheme is given in the following proposition:



\begin{proposition}\label{prop:multiplexing}
For FreeBee signals with co-prime intervals, folding for one signal restricts the fold sum of the others to a maximum of 1.

\end{proposition}

\begin{proof}
Let $T_1=x_1\Delta$ and $T_2=x_2\Delta$ be two beacon intervals where $x_1$ and $x_2$ are co-primes.
When a sampled RSSI vector containing beacons of interval $T_2$ is fold by $P=x_1$, beacon repeats in a column every LCM (least common multiple) of $x_1$ and $x_2$, which is $x_1\times x_2$. Since the total length of the sampled RSSI vector is much smaller than $x_1\times x_2$, beacons with the interval of $x_1\Delta$ cannot be folded into the same column when folded by $x_2$ (and vice versa), yielding the maximum fold sum of 1.
\end{proof}


\noindent The proposition states that the cross-interference between FreeBee signals is effectively suppressed when intervals are co-prime, essentially granting orthogonality between signals for concurrent communication. We note that this holds for both basic FreeBee and A-FreeBee.

In practice, to avoid the overhead in computing for co-prime numbers, we allow APs to select among a set of pre-stored \emph{prime} numbers instead. Moreover, the same interval should not be chosen among neighboring APs, which can be easily achieved in \emph{listen before talk} fashion; they listen to each other to acquire the beacon interval information (i.e., $x$ in $T=x\Delta$) that, according to the 802.11 standard, is recorded within beacons. After storing a set of intervals used by the neighboring APs, an unoccupied prime number is chosen as its interval. Conversely, coordination via wired connection (i.e., WLAN or the Internet) may be preferred, which avoids the hidden terminal problem.

\subsubsection{Implicit Addressing Feature of FreeBee}
This section discusses the unique addressing scheme used in FreeBee. As a reminder, demodulating each interval-multiplexed FreeBee signals remains the same as the case for a single signal; folding with $P$ yields the corresponding FreeBee signal with $P$.

\begin{figure}[h]
  \centering
  \includegraphics[width=0.48\textwidth]{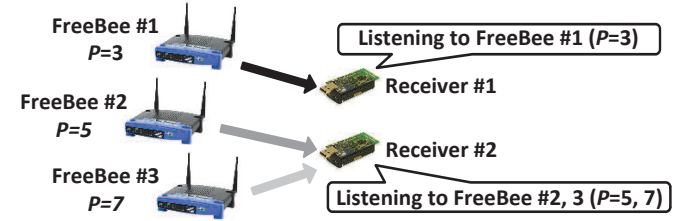}
  \caption{Implicit addressing via interval multiplexing}
  \label{fig:Implicit_Addressing}
\end{figure}

\noindent$\bullet$ \textbf{For stationary deployment}: As shown in Figure~\ref{fig:Implicit_Addressing}, each interval $P$ is allocated to one freeBee sender, implicitly addressing that sender. Hence a receiver may select a subset of $P$'s that corresponds to the sender(s) of interest.


\noindent$\bullet$ \textbf{For mobile deployment}: In mobile scenarios, mobile devices are not aware of $P$s used by nearby FreeBee enabled devices. We need to fold for all $P$s in the prime set to receive from every sender. This is not as computationally heavy as it may sound, as the number of primes, by the prime number theorem~\cite{PrimeTheory}, is limited to $\frac{x_{max}}{ln(x_{max})}-\frac{x_{min}}{ln(x_{min})}$ when $x_{max}$ and $x_{min}$ are the maximum and minimum values in the prime set. For example, there are 20 primes in the interval range of $x_{min}=53$ and $x_{max}=149$.

From a practical point of view, we emphasize that implicit addressing is a unique and effective feature of our design: as each sender is required to select different beacon interval, symbols demodulated with the same $P$ are ensured to be from the same sender. This allows safely constructing a long symbol by appending received symbols. However, this is not the case for all other cross-technology techniques. In Esense~\cite{ChebroluD09} and HoWiES~\cite{ZhangL13a}, a sender ID has to be embedded along with information in order to concatenate separate symbols correctly, leading to a large overhead in such a low-rate communication.

\subsection{Summary of Unique Design Benefits} \label{sec:free_channel}
Here we briefly summarize significant and unique benefits of FreeBee. Free side-channel design (i.e., not incurring additional traffic overhead) guarantees \textbf{(i) transparency} to avoid disrupting legacy networks, maintains \textbf{(ii) spectrum efficiency} by preventing competing for the spectrum, and effectively \textbf{(iii) supports mobile and duty cycled receivers}. Duty cycling is a critical technique to enable long-term operations~\cite{cao20122, LiLL14, INFOCOM15WRx} for battery-powered devices, where radios turned off for the majority of the time. FreeBee messages are continuously transferred over the air without overwhelming the channel, thereby safely reaching mobile and/or duty cycled receivers without the need for complex rendezvous techniques~\cite{Dutta08, HuangYZKH13}. Moreover, interval multiplexing enables \textbf{(iv) concurrent many-to-many communication} to further enhance the practicality under today's crowded wireless devices. Finally, FreeBee causes \textbf{(v) minimal computational overhead} to both sender (perturbing beacon transmission timing) and receiver (low-rate RSSI sampling of 7.8KHz and fold sum), easily affordable even to low-end embedded devices as demonstrated later in the paper. Above advantages are extensively evaluated, both analytically and experimentally, in Sections~\ref{sec:analysis} and~\ref{sec:evaluation}.

\section{Addressing Practical Issues} \label{sec:DiscussionChallenges}

This section discusses practical issues and their impact, followed by the solutions adopted in our design. 

\subsection{Reliability under Channel access delay}
Although beacons are prioritized over data packets and hence queueing delays are negligible~\cite{IEEE80211}, they \emph{do} suffer from channel access delays according to CSMA. This is in fact a challenge uniquely imposed on our design, which is different from traditional the PPM environment where all pulses are transmitted at their exact times. Upon long delays, a beacon may fail to contribute to folding. This is precisely why  beacon repetitions are needed (e.g., four  beacons in Figure~\ref{fig:Folding}) for statistical performance guarantee.  Our empirical study in  Section~\ref{sec:SER_Experiment} suggests that 5 beacons yield an error $<$ 1\%.


\subsection{Robustness to Noise}
Any non-beacon signal occupying the spectrum serves as noise and is a potential source of error. That is, as frequent 1's due to noise fill up the sampled RSSI, there is an increased chance of a large fold sum formed elsewhere to the beacon position, thus inducing demodulation failure. In other words, the performance of our design is enhanced by reducing the noise. This is simply done by taking only the first two RSSI samples for \emph{any} packet including beacons and discarding (i.e., set to 0) the rest. The reason behind this approach is two-fold: (i) As data packets tend to be much longer than beacons, this reduces noise to 1/6 on average in our experiments. (ii) Our empirical analysis indicates that the channel access delays of beacons are mostly ($>$90\%) less than 256us, where a similar result was reported in a recent study~\cite{HaoZXM11}. Noting that the duration of a RSSI sample is 128us, this suggests the first two RSSI samples (256us) of beacons maintain a high chance of overlapping (i.e., contributing to fold sum) upon folding.

\begin{figure}[h]
  \centering
  \includegraphics[width=0.46\textwidth]{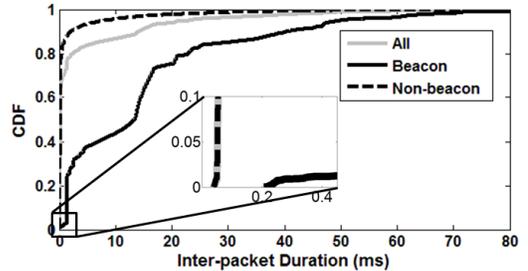}
  \caption{Distribution of inter-packet duration for 5,000 beacons and 40,000 non-beacon packets in university building with 50+ APs. More than 99\% of the beacons have more than 256us channel vacancy before their transmission.}
  \label{fig:Inter_Packet_Duration_CDF}
  \vspace{-2mm}
\end{figure}

We note that the aforementioned technique of capturing the first two RSSI is feasible only when the \emph{beginning} of beacon can be appropriately detected. Noting that WiFi inter-packet duration can be as small as 10us (SIFS in 802.11b/g/n), can RSSI with the coarse granularity of 128us successfully tell when the beacon starts, via transition from 0 to 1? To answer this question we performed an experiment under a dense WiFi environment with more than 50 APs deployed. The result depicted in  Figure~\ref{fig:Inter_Packet_Duration_CDF} demonstrates that inter-packet durations preceding beacons are predominantly ($>$ 99\%) larger than 256us. This is because a beacon is a single packet, unlike data traffic that are often large and cause multiples of bursty packets. Therefore the beginning of beacons can be safely captured via RSSI transition (0 $\rightarrow$ 1), validating the efficacy of the technique.

\section{Analytics} \label{sec:analysis}
This section provides a theoretical analysis of the performance of FreeBee and A-FreeBee under different settings. 

\subsection{SER vs. Sampling Duration} \label{sec:symbol_error_rate}



%
%

\begin{figure}[h]
\centering
\subfigure[\small{FreeBee}]{\label{fig:SERvsSampleDur_Sync}
\includegraphics[width=0.22\textwidth, height=0.15\textheight]{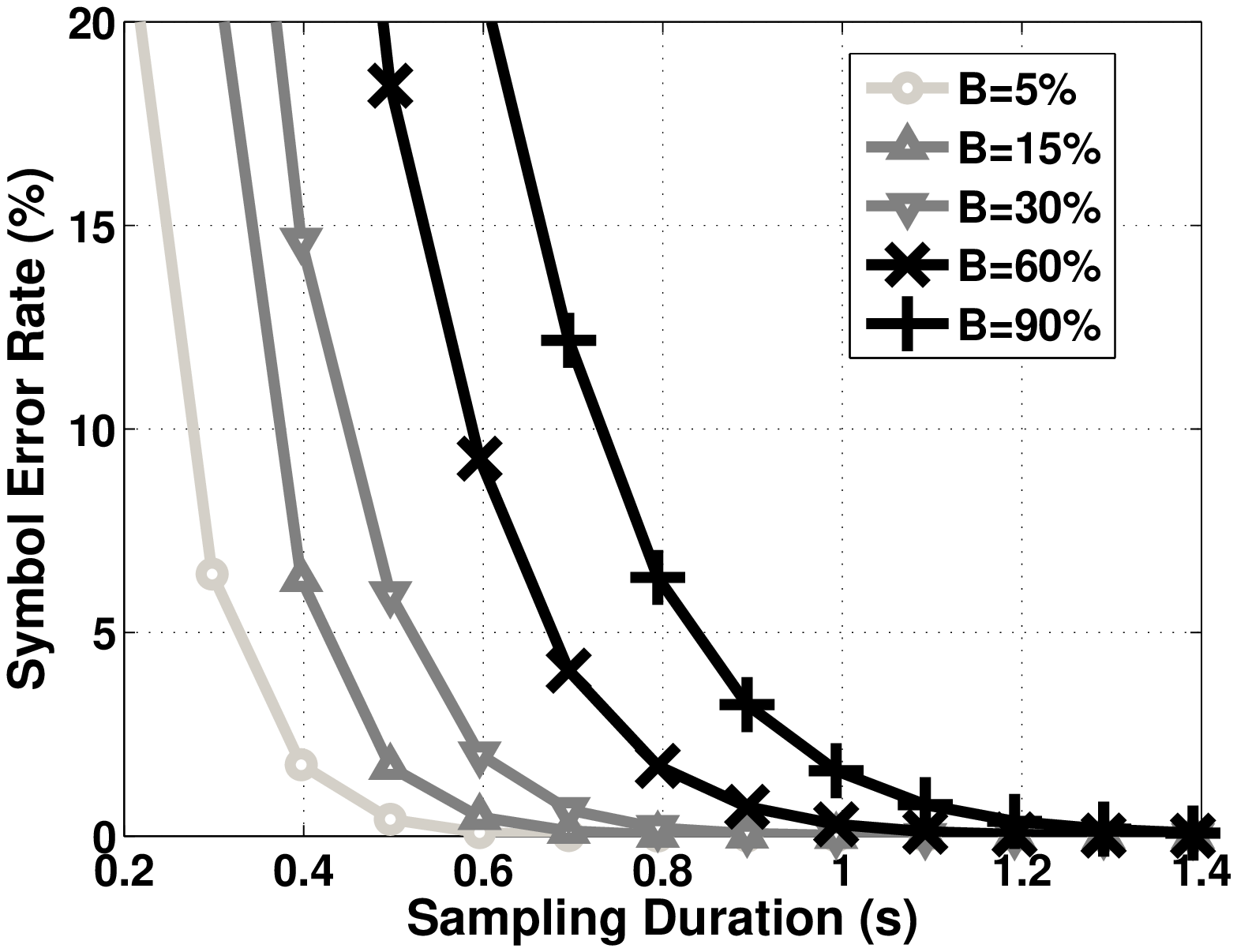}}
\subfigure[\small{A-FreeBee}]{\label{fig:SERvsSampleDur_Async}
\includegraphics[width=0.22\textwidth, height=0.15\textheight]{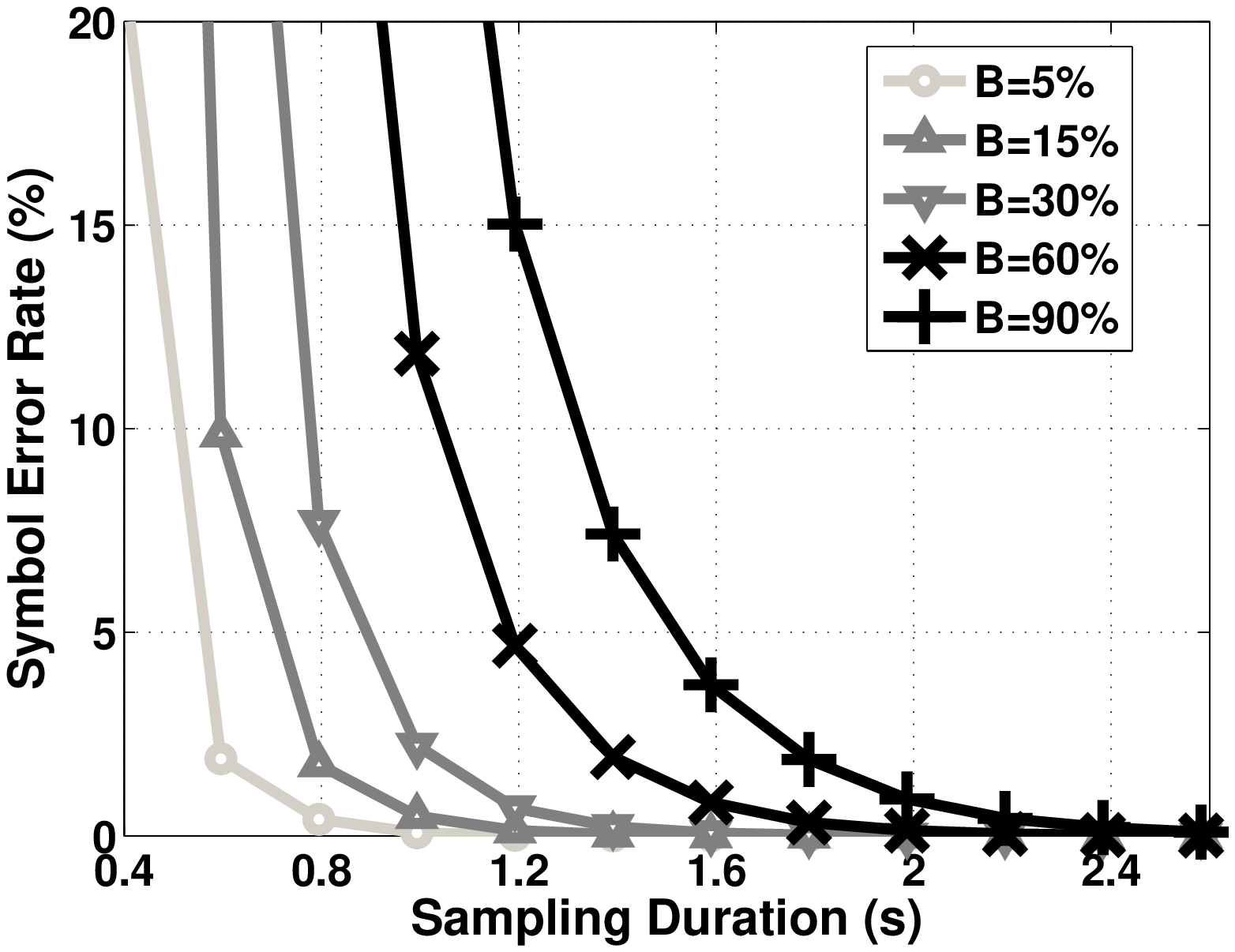}}
\vspace{-0.1in}
\caption{SER vs. sampling duration ($T=97\Delta$).}
\label{fig:SERvsSampleDur}
\vspace{-0.1in}
\end{figure}

As noted earlier, the symbol error occurs upon demodulation failure, essentially when the fold sum of noise is higher than that of the beacons. For brevity, we omit a detailed SER (Symbol Error Rate) derivation which can be found in our work of~\cite{songminFREEBEE}, and move directly to the results to demonstrate the impact of three  system parameters: (i) $T $, the beacon interval. (ii)  $\rho$, the number of beacon repetitions for statistical demodulation; and (iii) Sampling duration,  the sampling time to obtain a symbol, which is  $\rho\times T$ for FreeBee and $\rho\times 2T$  for A-FreeBee.

Figures~\ref{fig:SERvsSampleDur_Sync} and~\ref{fig:SERvsSampleDur_Async} show Symbol Error Rate (SER) when default beacon interval is set as $T=97\Delta=99.3ms$.  These figures convey three ideas: (i) longer sampling (i.e., higher $\rho$) lowers SER, as more beacons are utilized to fight against noise; (ii) for a given duration, FreeBee achieves a lower SER than A-FreeBee; and (iii) higher channel occupancy rate, denoted by $B$, indicates more noise, thus higher SER. The figures have $B$ up to 90\% for completeness of analysis, where $B\leq 30\%$ was observed in our experiment (under the threshold of -75dBm) in a university building with 50+ APs. Similar observations were reported in a recent study from a large-scale open WiFi trace~\cite{Crawdad, ZhouXXSM10}. The figures suggest analytically that FreeBee and A-FreeBee achieve SER$<$1\% for durations of 0.7$s$ and 1.2$s$ under such condition.  We note this analytical result matches well with results from empirical experiments shown in Section~\ref{sec:SER_Experiment}.

\begin{figure}[h]
\vspace{1.5mm}
  \centering
  \includegraphics[width=0.42\textwidth]{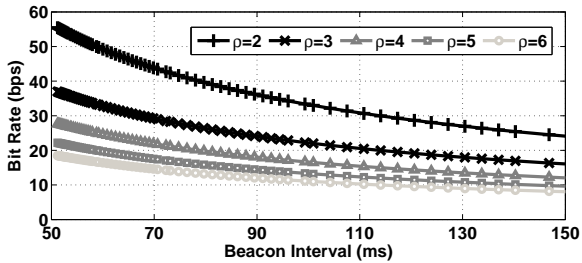}
  \vspace{-2mm}
  \caption{Bit rate vs. $T$ for different $\rho$.}
  \label{fig:BitRate}
 \vspace{-4mm}
\end{figure}
\vspace{-0.05in}

\subsection{Transmission Rate vs. Beacon Interval} \label{sec:transmission_rate}

Beacon interval, $T$, is another factor that affects the performance. The impact of beacon interval $T$ can be observed with bit rate $R$. Intuitively, enlarging the beacon interval has two effects: (i) it offers more space for shift, or equivalently, yields more bits per symbol; and (ii) it requires more time to reach the same $\rho$. The bit rate for FreeBee can be computed as below:

\vspace{-0.1in}
\begin{equation}
 \label{eq:bit_rate}
 R=\frac{log_2T/\Delta}{T\times \rho}\:bps
\end{equation}
\vspace{-0.1in}

\noindent
Noting that $\Delta$ (1.024ms in WiFi) defines the granularity of shift,
the numerator in the Equ.~\ref{eq:bit_rate} implies $bit\:per\:symbol$. Figure~\ref{fig:BitRate} shows the impact of beacon interval $T$ on $R$ in different scenarios for the range of practical intervals. In A-FreeBee, the rate is cut in half as it takes double sample duration (i.e., $\rho\times 2T$) to convey a same symbol. It is important to note that the rate given here is per sender, without bandwidth consumption (i.e., without incurring extra traffic). Due to interval multiplexing, the aggregated throughput linearly increases according to the number of senders. Furthermore, boosting the throughput of a single sender by adopting additional beacons is also a viable option. Performances under such cases are shown via experimental evaluations in the next section.

\section{Performance Evaluation} \label{sec:evaluation}

We show the generality of (A-)FreeBee by evaluations performed across four platforms operating on three different wireless standards: WiFi, ZigBee, and Bluetooth.

\begin{figure}[h]
 \centering
\begin{minipage}[t][0.33\textheight]{0.47\textwidth}
 \centering
 \subfigure[\small{WiFi(WARP)}]
 {\label{fig:Testbed_WARP}
  \includegraphics[width=0.47\textwidth, height=0.13\textheight]{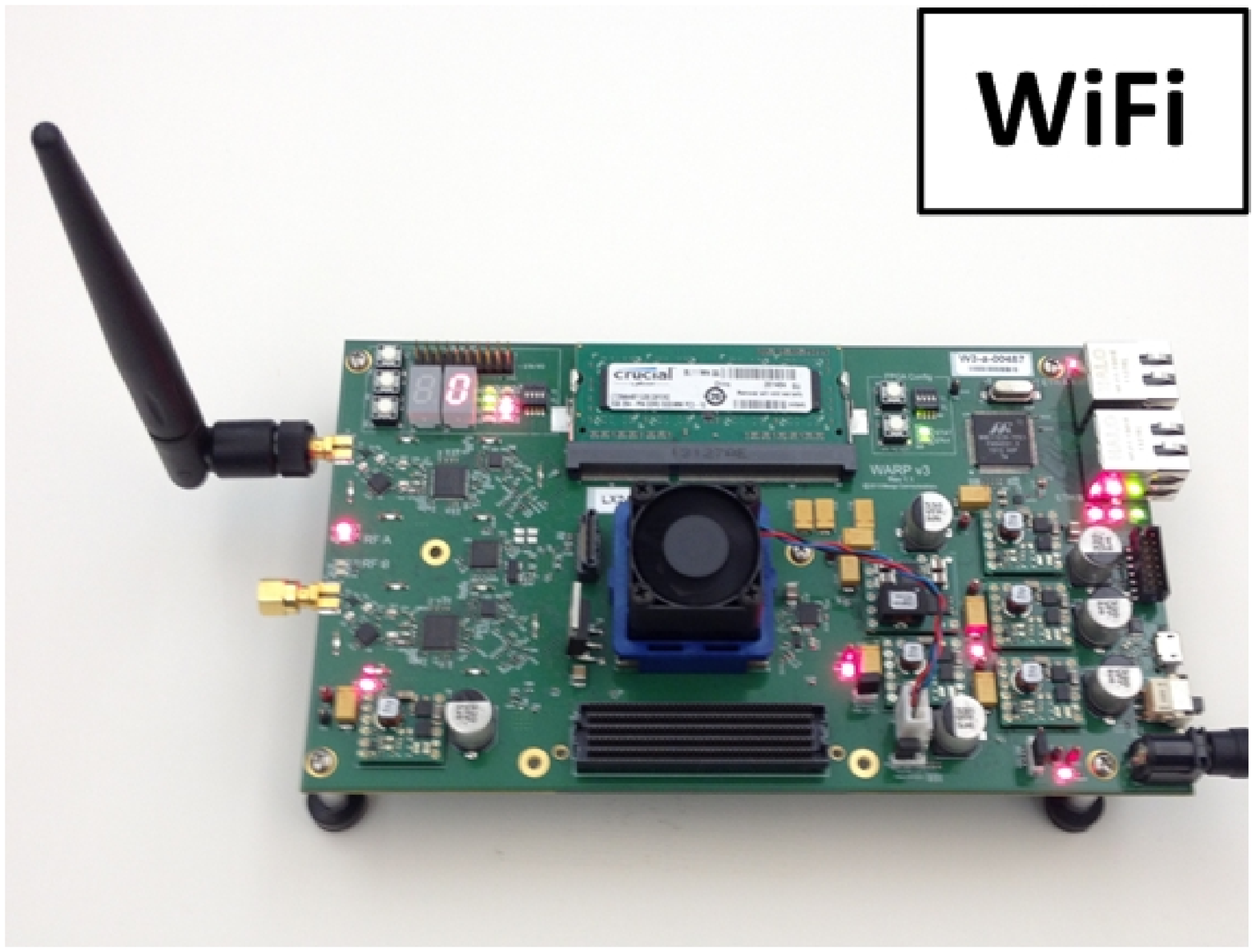}}
 \subfigure[\small{WiFi(laptop)}]
 {\label{fig:Testbed_Laptop}
 	\includegraphics[width=0.47\textwidth, height=0.13\textheight]{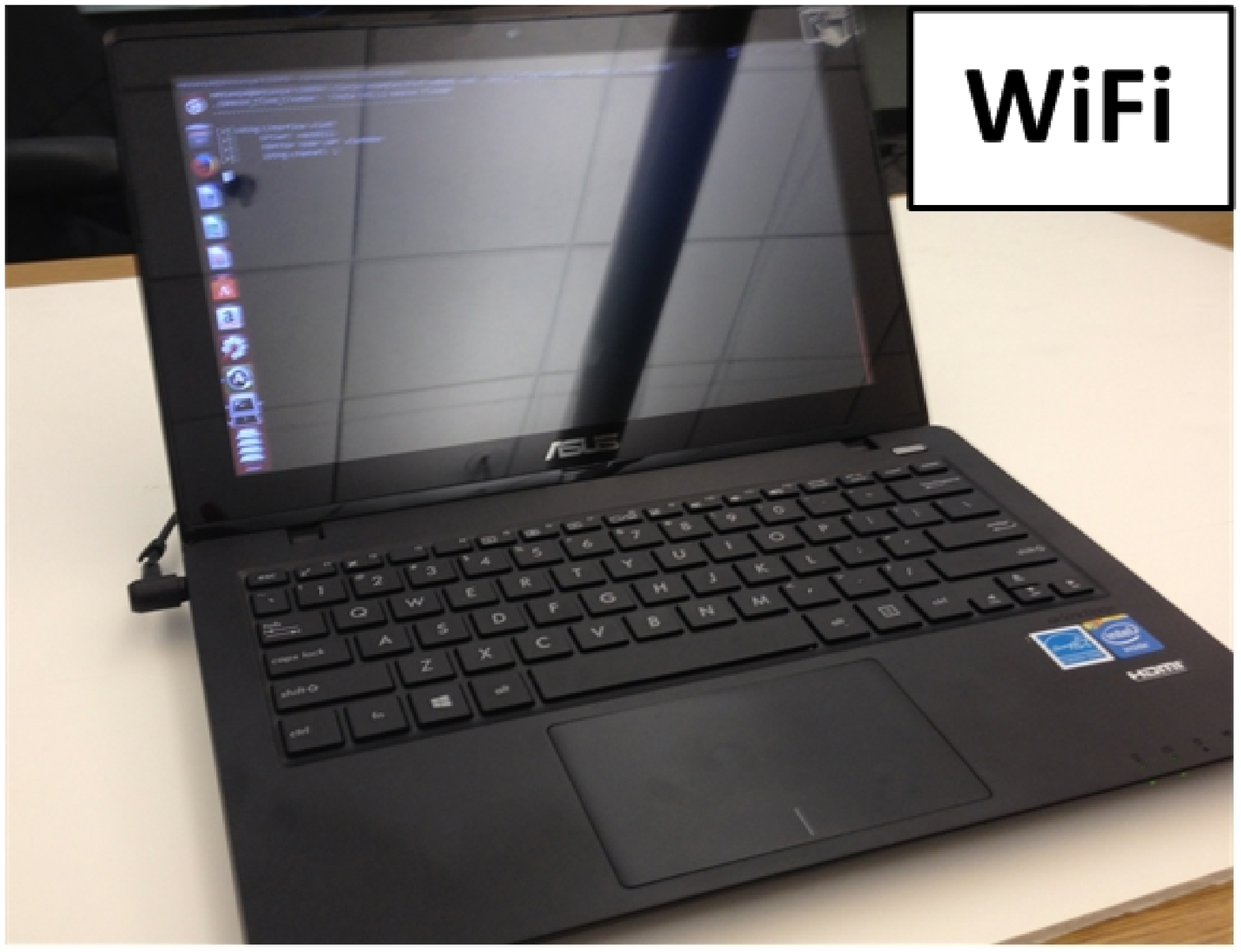}}
 \subfigure[\small{ZigBee}]
 {\label{fig:Testbed_MICAz}
		\includegraphics[width=0.47\textwidth, height=0.13\textheight]{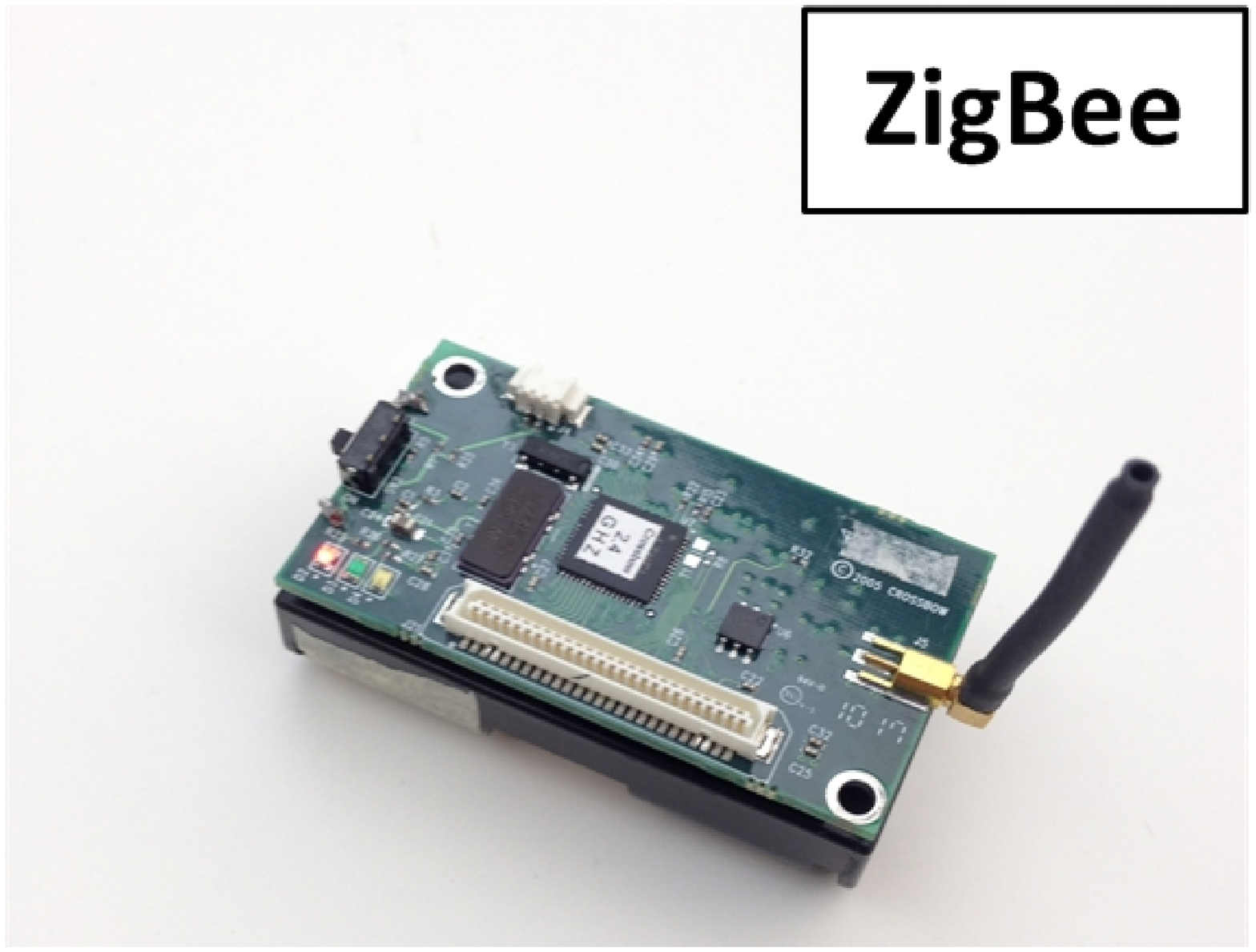}}
 \subfigure[\small{Bluetooth}]
 {\label{fig:Testbed_Bluetooth}
  \includegraphics[width=0.47\textwidth, height=0.13\textheight]{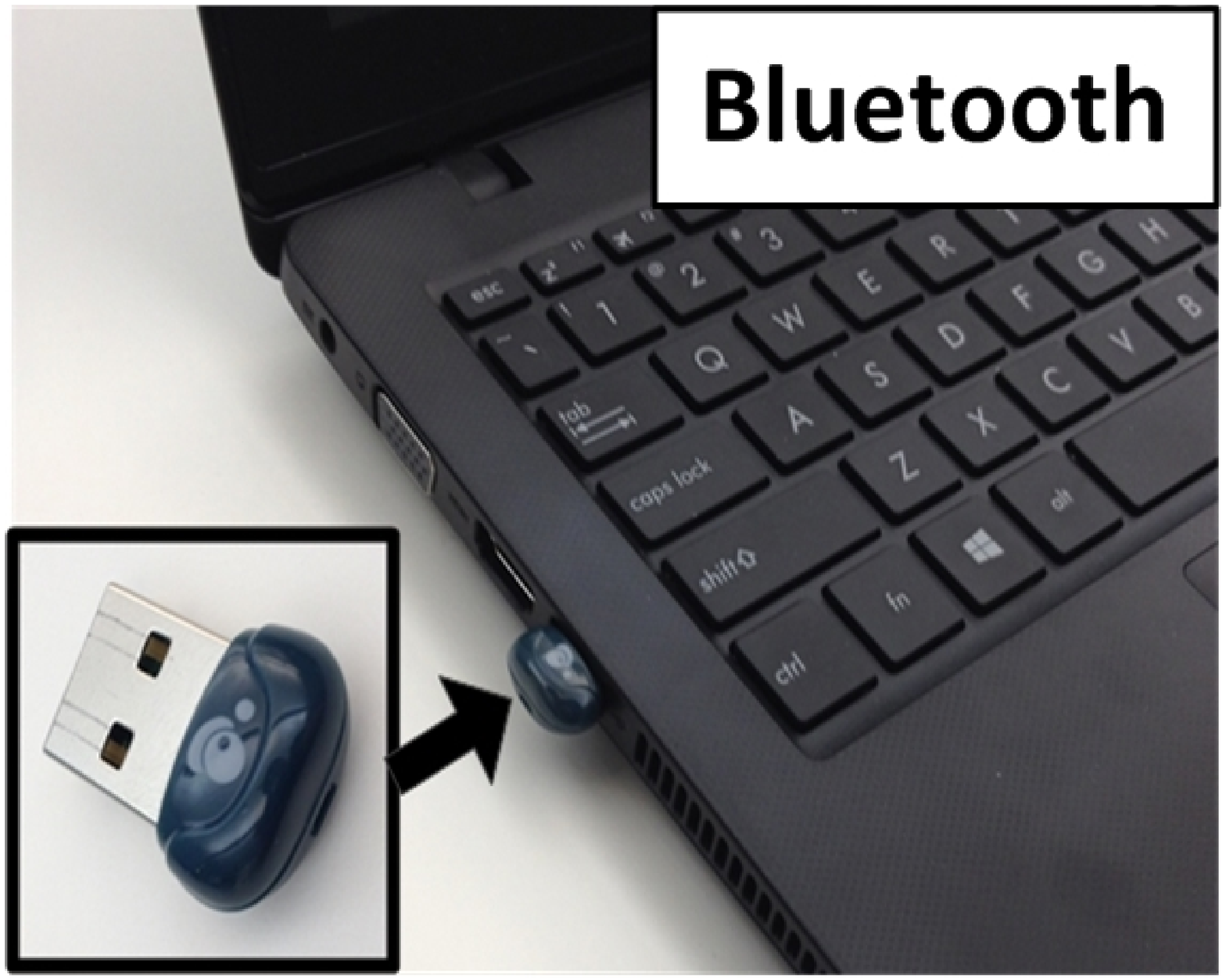}}		
 \caption{Four evaluation platforms on three different wireless technologies -- WiFi, ZigBee, and Bluetooth.}
 \label{fig:Testbed_platforms}
\end{minipage}
\end{figure}

\begin{table}[h]
\small{
\caption{Experiment Settings}
\label{table:Experiment_Scenario}
\begin{tabular}{|c||c|c|c|c|}
\hline
\small{\textbf{Communication}} & \small{\textbf{Tx}} & \small{\textbf{Tx}} & \small{\textbf{Rx}} &   \\
\small{\textbf{Direction}} & \small{\textbf{Ch.}} & \small{\textbf{Power}} & \small{\textbf{Ch.}} & \small{\textbf{Dist.}} \\ \hline \hline
\small{\textbf{WiFi $\rightarrow$ ZigBee}} & 1,4 & 13 dBm & 11-15 & 8m \\ \hline
\small{\textbf{ZigBee $\rightarrow$ WiFi}} & 11-14 & 0 dBm & 1 & 1.5m \\ \hline
\small{\textbf{Bluetooth $\rightarrow$ WiFi}} & 37-39 & 4 dBm & 4 & 3m \\ \hline
\small{\textbf{Bluetooth $\rightarrow$ ZigBee}} & 37-39 & 4 dBm & 15 & 3m \\ \hline
\end{tabular}
}
\end{table}
\vspace{-0.1in}

\begin{figure*}[t]
\centering
\begin{minipage}[t][0.21\textheight]{1\textwidth}
\centering
\subfigure[\small{FreeBee}]{\label{fig:Eval_SER_FreeBee} \includegraphics[width=0.3\textwidth, height=0.15\textheight]{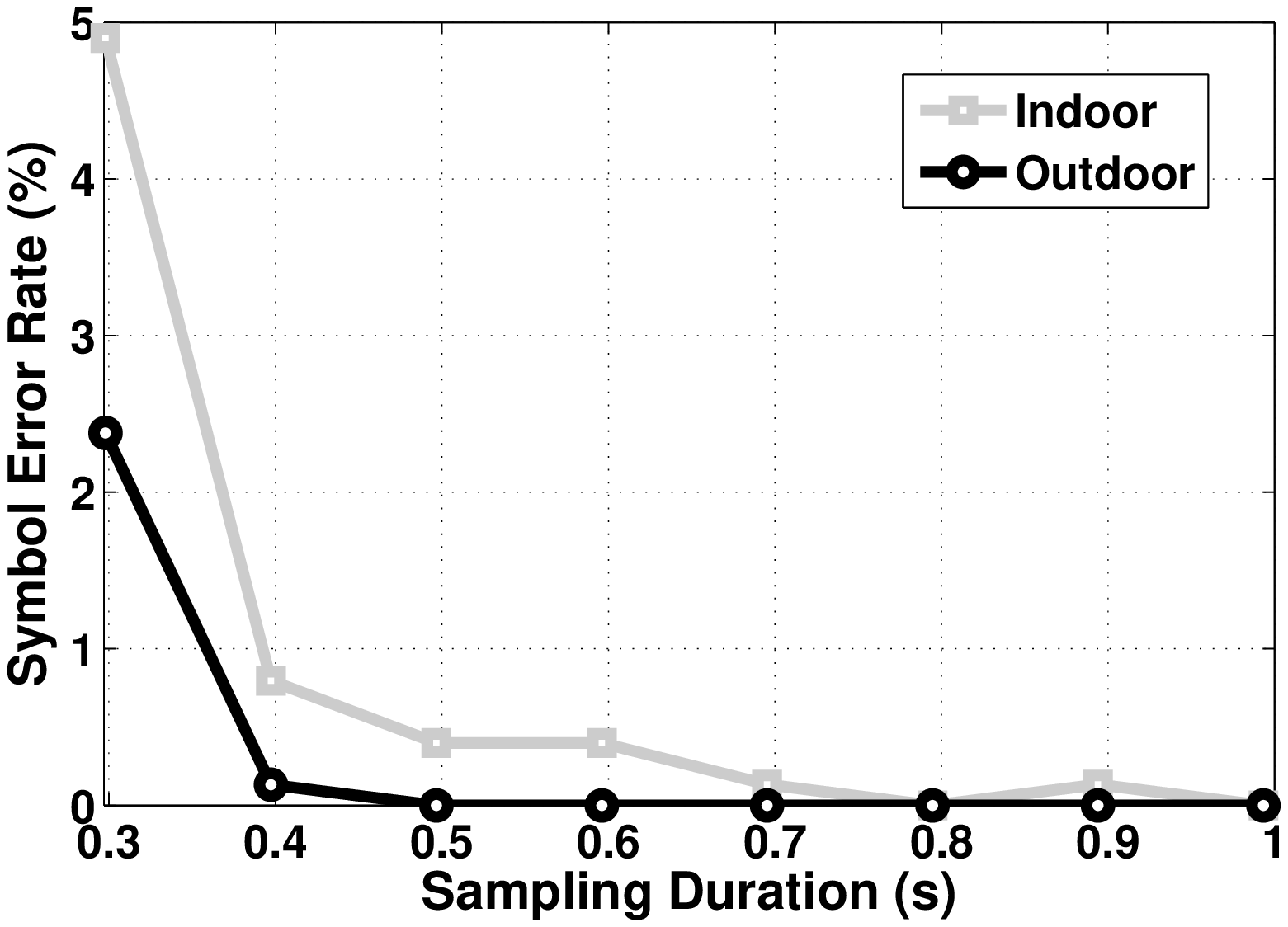}}
\hspace{0.5mm}
\subfigure[\small{A-FreeBee}]{\label{fig:Eval_SER_AFreeBee} \includegraphics[width=0.3\textwidth, height=0.15\textheight]{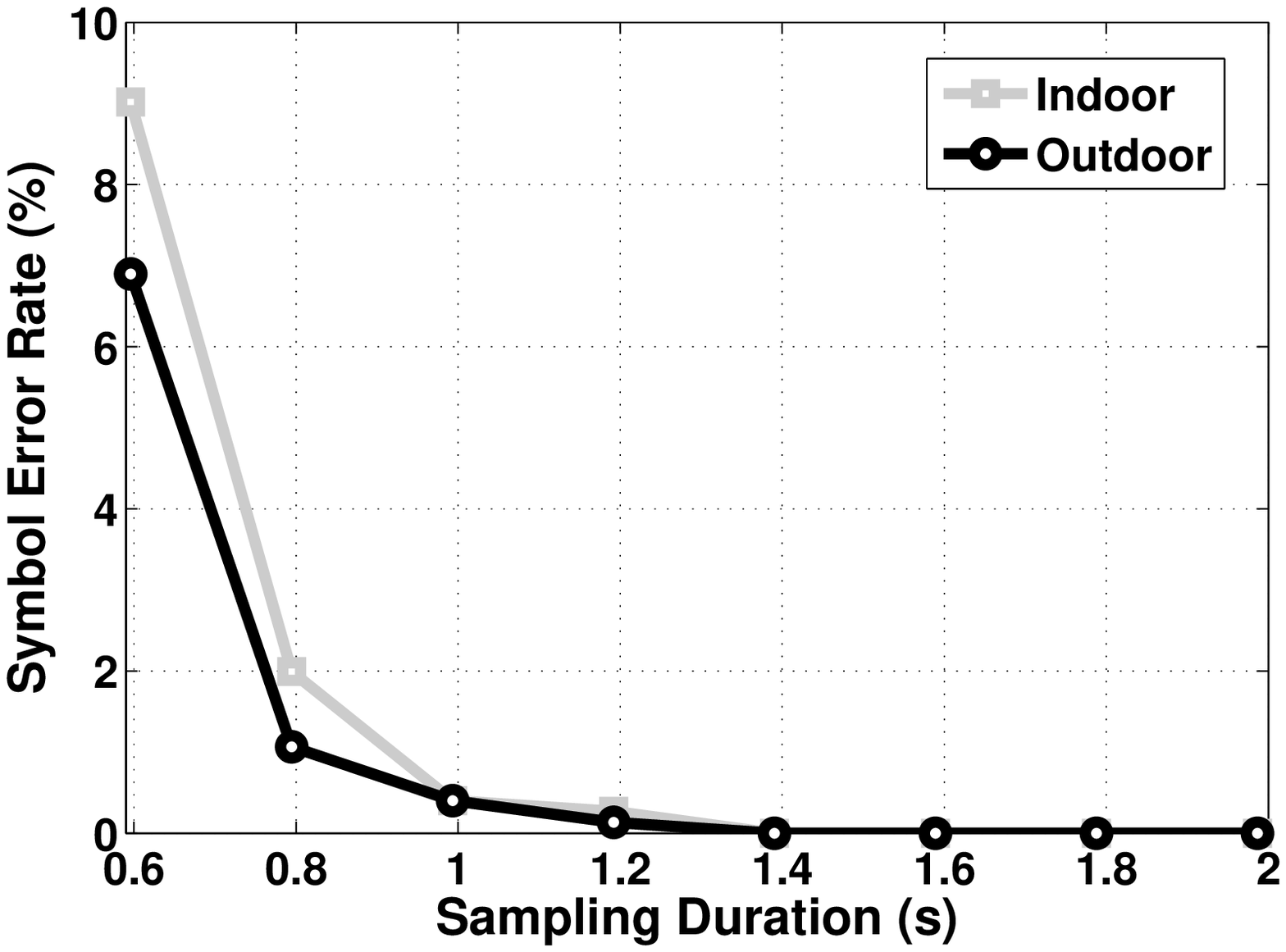}}
\hspace{0.5mm}
\subfigure[\small{Constellation diagram and distribution}]{\label{fig:Eval_Constellation} \includegraphics[width=0.33\textwidth, height=0.15\textheight]{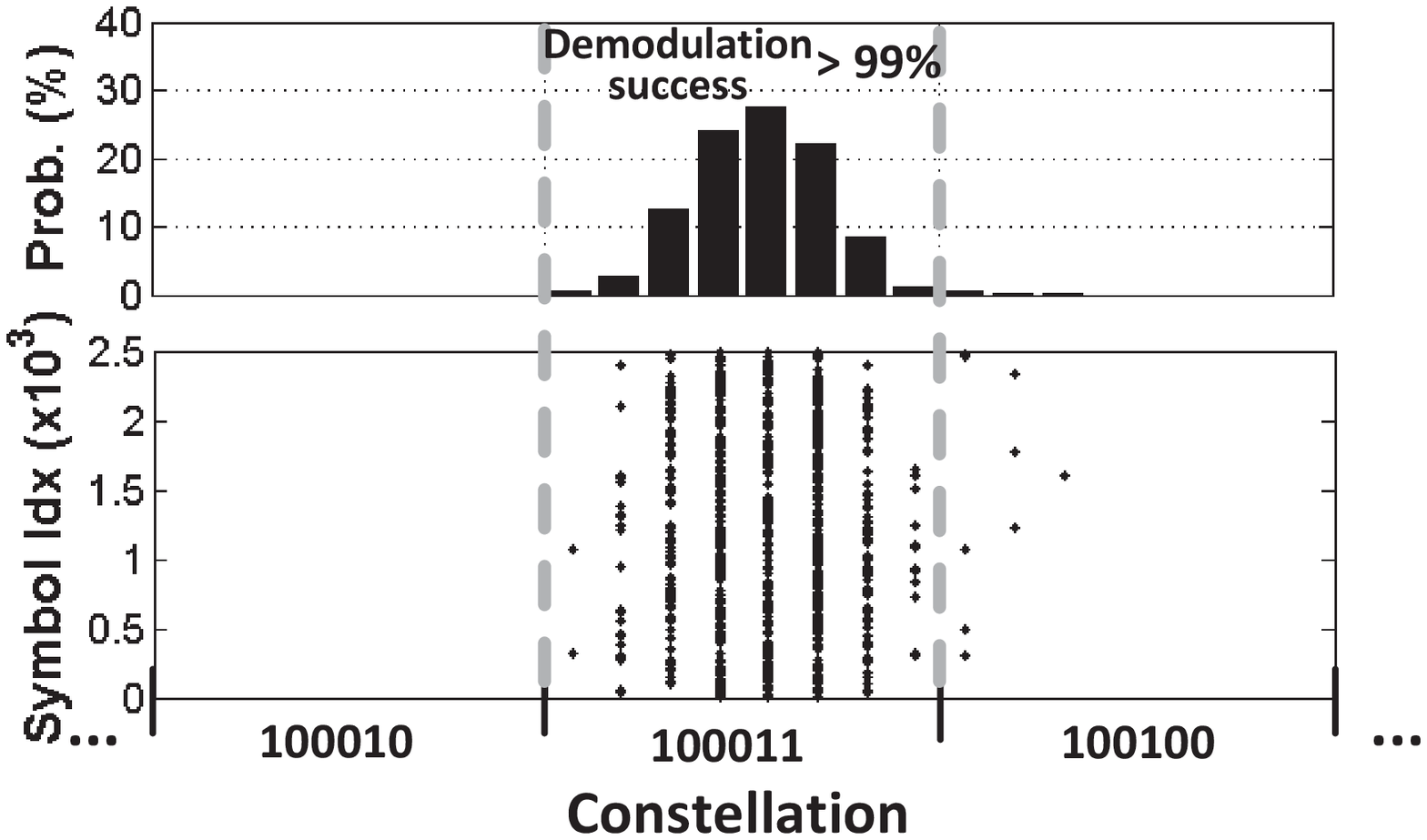}}
\caption{Indoor and outdoor performances of (a)FreeBee and (b)A-FreeBee when $T=97\Delta=99.3ms$. Due to less noise outdoors, both FreeBee and A-FreeBee reach SER$<$1\% at $\rho=4$, whereas indoor environment required $\rho=5$. (c) shows where the received symbols are positioned when `100011' was sent 2,500 times.}
\label{fig:SER_evaluation}
\end{minipage}
\end{figure*}

\subsection{Experiment Settings}
The experiment parameters are specified in Table~\ref{table:Experiment_Scenario}, where the detailed settings are as follows.

\vspace{1mm}
\noindent$\bullet$\textbf{ Design generality:}
As a generic cross-technology communication framework, our design mechanisms including mo-dulation, demodulation, and interval multiplexing commonly apply to different underlying technologies. This is possible because, according to the standards, (i) beacons are adopted in all three technologies of  WiFi, ZigBee, and Bluetooth, (ii) they commonly allow modification of beacon timings, (iii) they reside on overlapping frequencies in the 2.4GHz band and finally, (iv) the light-weight design makes our design feasible even to low-end devices, as demonstrated later in the section.

\vspace{1mm}
\noindent$\bullet$\textbf{ WiFi:} Figures~\ref{fig:Testbed_WARP} and~\ref{fig:Testbed_Laptop} show the two WiFi platforms on which our design is implemented: WARP~\cite{WARPv3} and laptops. The former is a wireless research platform fully integrated with WiFi PHY/MAC layers. As a FPGA-based system allowing real-time operation, the evaluation on WARP enables precise exploration into FreeBee performance.
Further implementation on six different laptops with various WiFi NICs from Qualcomm, Realtek, and Intel avoids hardware bias.

Evaluations via laptops utilize \texttt{Lorcon2} packet injection library~\cite{Lorcon2} to schedule beacons according to FreeBee, which is a reasonable approach since laptops/PCs running software AP are frequently used in practice~\cite{hostapd}. In Table~\ref{table:Experiment_Scenario}, we use WiFi channel 1 (unless otherwise specified) for communication with ZigBee, one of the three most commonly used channels (i.e., 1, 6, and 11), granting generality to our result. We then tune to channel 4, which overlaps with a Bluetooth advertising channel (i.e., 38), for communication with Bluetooth.

\vspace{0.1in}
\noindent$\bullet$\textbf{ ZigBee:} We use 30 ZigBee-compliant MICAz nodes (Figure~\ref{fig:Testbed_MICAz}) for our experiments.
When operating as a receiver, a MICAz node captures RSSI values (at 7.8KHz sampling rate by default) and records them within its 512KB on-board flash. The values are either processed (i.e., demodulated) within the node or flushed to a PC for analysis, depending on experiments. We use channels 11-15, overlapping with WiFi channels 1 and 4, and a Bluetooth advertising channel of 38.

\vspace{0.1in}
\noindent$\bullet$\textbf{ Bluetooth:} Our design is implemented on IOGEAR Bluetooth LE USB adapter, a cheap ($\sim$12 USD) off-the-shelf product, shown in Figure~\ref{fig:Testbed_Bluetooth}. On this device, we utilize \texttt{AltBeacon}~\cite{AltBeacon} library running under Linux's BlueZ driver for FreeBee embedding. Specifically, connectable directed advertising was used to generate FreeBee-enabled beacons
on all three advertising channels of 37-39, which complies to the standard on Bluetooth beaconing.


\subsection{Symbol Error Rate}\label{sec:SER_Experiment}
Here we present the reliability of our design in practice by evaluating SER under both indoor and outdoor environments in a residential area. This experiment is based on WiFi(WARP) to ZigBee communication, where more than 2,500 symbols are sent and demodulated for SER computation. The results for FreeBee and A-FreeBee are depicted in Figures~\ref{fig:Eval_SER_FreeBee} and~\ref{fig:Eval_SER_AFreeBee}. Both designs reach $SER \leq 0.5\%$ when $\rho=5$ (i.e., $0.5s$ and $1s$ for FreeBee and A-FreeBee), regardless of the environment. Furthermore, both designs perform better in outdoor environments, due to lower channel occupancy, $B$.

The lower figure in~\ref{fig:Eval_Constellation} illustrates the constellation, along with the demodulated positions of received symbols, when the 6bit symbol of `100011' is repeatedly sent 2,500 times. Demodulation is successful when a dot resides inside the region marked by gray dotted lines in the constellation. The upper figure in~\ref{fig:Eval_Constellation} shows the distribution of the dots, in which $>$99\% are successfully demodulated.

FreeBee essentially trades-offs speed with robustness, by controlling $\rho$; The increase of $\rho$ effectively improves the chance of detecting the peak at the position of the beacon and thus lowers symbol error rate (i.e., better robustness), at the cost of increased time for symbol delivery (i.e., slower speed). To better demonstrate this feature, we conduct an experiment study under peak interference of $B=30\%$ in a university building . In our study, the degree of interference is occasionally observed during peak times of the day, in a university building with 50+ WiFi APs. Under this setting, our experiments demonstrate SER of 3.1\%, 1.8\%, and below 1\% when $\rho=$ 13, 14, and 15, respectively. This result, compared to the performance under usual traffic (Figure~\ref{fig:SER_evaluation}(a)), clearly delivers the idea that FreeBee can be effectively achieved by increasing $\rho$ even under extreme interference, where the throughput may decrease to as lows as 30\%.

\begin{figure}[h]
  \centering
  \includegraphics[width=0.47\textwidth]{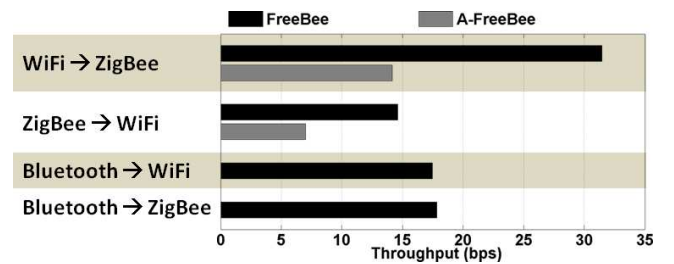}
  \vspace{-2mm}
  \caption{Per-sender throughputs achievable without incurring any extra traffic, evaluated on four different communication channels}
  \label{fig:Evaluation_Throughput_FreeChannel}
  \vspace{-3mm}
\end{figure}

\subsection{Per-sender Throughput}\label{sec:PerSenderEval}
In this section we demonstrate the data rate achievable per sender, under free side-channel as well as when the channel is fully utilized.

\noindent$\bullet$\textbf{ Free side-channel throughput:} Figure~\ref{fig:Evaluation_Throughput_FreeChannel} illustrates per-sender data rate without introducing additional traffic (i.e., free side-channel). The experiment was conducted in a residential area with 20+ APs in proximity. Beacon intervals are set as 99.3, 76.8, and 78.75ms for WiFi, ZigBee, and Bluetooth, respectively, where the rate may easily be enhanced with shorter beacon intervals. The figure demonstrates two ideas: (i) The performance of A-FreeBee is slightly less than half of that of FreeBee, due to doubled sampling duration in A-FreeBee. Longer sampling duration increases the chance of larger fold sum of noise, hence yields higher SER in A-Freebee compared to FreeBee. This agrees with our theoretical analysis in Section~\ref{sec:symbol_error_rate}. (ii) Among different communication directions, WiFi to ZigBee exhibits the fastest rate of 31.5bps for FreeBee. The rate drops to 14.6bps for ZigBee to WiFi as the ZigBee standard enforces large unit shift (i.e., $\Delta$=15.36ms), reducing the amount of information embedded in a symbol. While the Bluetooth standard defines a fine-grained beacon shift unit ($\Delta$ = 0.625ms), the random backoff ranging up to 10ms affects the performances of Bluetooth to WiFi and ZigBee communications where they show 17.5 and 17.8bps, respectively. While disabling the backoff functionality would increase the throughput significantly, the case is not considered in this paper as it requires modification to the standard (lacks compatibility).




\noindent$\bullet$\textbf{ Upper bound throughput under ideal conditions:} We demonstrate the effectiveness of our design by comparing the maximum throughput of our design to that of Esense~\cite{ChebroluD09}, a state-of-art cross-technology communication scheme. Esense, as a WiFi to ZigBee communication technique, conveys data by modulating WiFi packet durations, where its ideal maximum throughput is reported under the ideal condition of interference-free channel that is fully utilized by a single sender. Hence we adopt this setting in evaluating FreeBee. Evaluation parameters also follow the values proposed in Esense; That is, RSSI sampling rate of 32$KHz$ (= sampling interval of 30.5$us$), inter-frame duration of $90us$, and the maximum WiFi transmission rate of 54Mbps (802.11g).

To obtain the ideal FreeBee performance, we first note that beacons need not be repeatedly transmitted under the interference-free channel.
That is, under zero-interference ideal environment, any energy detected in the channel (via RSSI) is ensured to be a beacon, where the symbol is simply decided by the timing of the RSSI sample. In other words, a single beacon frame conveys a symbol where SER=0. When the maximum shift is $x\Delta$, a symbol (i.e., a beacon) embeds $log_2(x+1)bits$, where the time it consumes consists of inter-frame duration, beacon transmission time, and the amount of shift, where the unit of shift may be as small as the sampling interval (i.e., $\Delta = 30.5us$). When $x=4$, FreeBee yields the throughput of 10.2Kbps given the beacon length of 100Bytes. Esense, according to its researchers, achieves the throughput of 5.13Kbps under the same setting. This is because Esense requires to use long-length packets (with long air-time) up to 1,500Bytes to enable measurement of WiFi packet durations via a low-end ZigBee node. Meanwhile, FreeBee utilizes short beacons to offer higher channel efficiency.




\begin{figure}[h]
  \centering
  \includegraphics[width=0.48\textwidth]{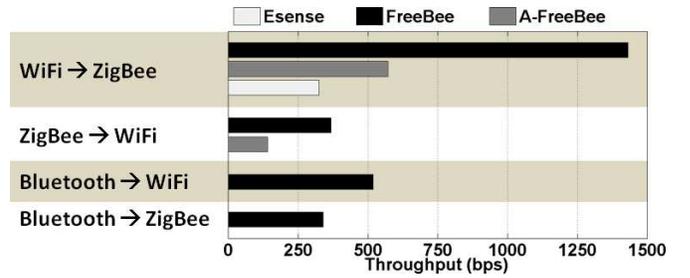}
  \caption{Aggregated throughput via interval multiplexing across multiple senders and/or within a sender. Compared to the state-of-the-art,
Esense, FreeBee and A-FreeBee shows 4.3$\times$ and 1.8$\times$ performance, respectively.}
  \label{fig:Evaluation_Throughput_DominatedChannel}
\vspace{-3mm}
\end{figure}

\subsection{Aggregated Throughput}
This section investigates throughputs achievable via interval multiplexing, under different communication scenarios. Through interval multiplexing our design supports: (i) boosted aggregated throughput via concurrent transmission from multiple senders and (ii) adoption of additional beacons in a single sender to increase the throughput of that sender. We note that (i) and (ii) can be applied simultaneously; that is, multiple senders transmit concurrently while each sender freely adopts additional beacons in case of a need for higher throughput. Figure~\ref{fig:Evaluation_Throughput_DominatedChannel} depicts the throughputs for different communication scenarios, where we compare our design to the state of the art, Esense. According to the authors, when multiple Esense senders are present and the interference from the heterogeneous wireless systems are negligible, Esense achieves the bit rate of 1.63Kbps when a single WiFi packet is used to deliver one symbol, where five consecutive packets are needed per symbol for reliable communication~\cite{ChebroluD09}. This yields a throughput of 326bps. Following the same setting, Figure~\ref{fig:Evaluation_Throughput_DominatedChannel} demonstrates the rate achievable when the channel is shared among multiple (A-)FreeBee senders. FreeBee from WiFi to ZigBee shows the highest throughput of 1.4Kbps. The rate drops in other scenarios due to the same reasons discussed in the previous section; that is, a large shift unit in ZigBee and random access delay in Bluetooth. It is notable that Bluetooth to WiFi shows higher throughput of 514bps, compared to 349bps in the Bluetooth to ZigBee case. This is because higher sampling rate achievable in WiFi compared to ZigBee (10MHz versus 7.8KHz in our setting) allows precise measurement of the beacon timings and reduces the chance of noise forming a high fold sum.

We note that the results in Figure~\ref{fig:Evaluation_Throughput_DominatedChannel} is achievable under a large number ($\sim$hundreds) of concurrent senders, which can be realized by the following two techniques: (i) To allow numerous interval-multiplexed symbols, we choose beacon intervals from a set of numbers whose pair-wise LCM is longer than the length of the sampled RSSI vector. This clearly offers many more adoptable intervals compared to just using primes, while maintaining the effect of the interval multiplexing (i.e., Proposition~\ref{prop:multiplexing} holds). (ii) As shown in the previous section, interval-multiplexed FreeBee symbols are transparent to each other under usual volume of concurrent senders (e.g., tens of FreeBee WiFi APs). Here we consider the extreme case of hundreds of senders where we take cross-interference into account. We address this issue of cross-interference by eliminating the beacons from the RSSI vector after they are interpreted. Specifically, we demodulate symbols in the ascending order according to their intervals since the FreeBee senders with shorter beacon interval have more RSSIs and those RSSIs are interference from the perspective of other senders.
Demodulating symbols in an ascending order according to the beacon interval maximally eliminates interference to the symbols from other senders and thus improve the throughput.
For instance, in WiFi to ZigBee FreeBee, SER is suppressed to be less than 5\% when over 400 symbols with intervals between 91.1ms and 2s are concurrently sent.

\subsection{Cross-technology/channel Broadcast}
This section demonstrates FreeBee's unique capability to \emph{broadcast} to receivers with heterogeneous technologies and channels.

\begin{figure}[h]
  \centering
  \includegraphics[width=0.43\textwidth]{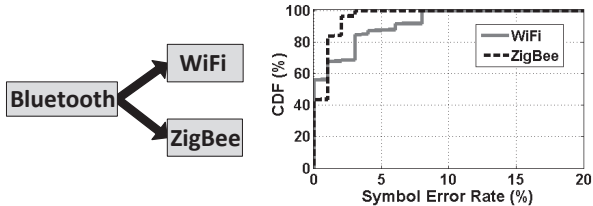}
  \vspace{-2mm}
  \caption{Simultaneous broadcast from Bluetooth to WiFi and ZigBee.}
  \label{fig:Cross_Technology_Broadcast}
\end{figure}

\noindent$\bullet$ \textbf{Cross-technology broadcast}: Depicted in Figure~\ref{fig:Cross_Technology_Broadcast}, as a generic communication framework, FreeBee allows broadcast to heterogeneous receivers with overlapping frequencies.  Our generic design is effective in practice, since it avoids the complexity associated with technology-specific operation. In the experiment, WiFi and ZigBee were set to channels 4 and 15 to listen to Bluetooth's advertisement channel of 38 simultaneously. In this particular case shown in Figure~\ref{fig:Cross_Technology_Broadcast},  WiFi receiver, compared to ZigBee, suffers from larger SER. It is due to higher noise, as WiFi channel 4 overlaps with the popular WiFi channels of 1 and 6, while ZigBee channel 15 does not.

\begin{figure}[h]
  \centering
  \includegraphics[width=0.43\textwidth]{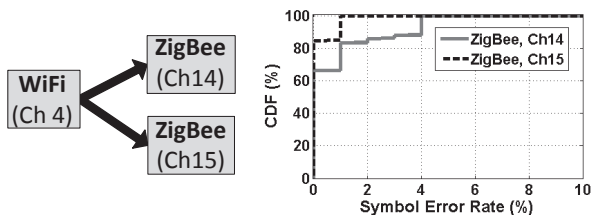}
  \vspace{-2mm}
  \caption{Simultaneous broadcast from WiFi to ZigBees on different channels.}
  \label{fig:Cross_Channel_Broadcast}
\end{figure}

\begin{figure*}[t]
 \centering
\begin{minipage}[t][0.19\textheight]{1\textwidth}
 \centering
 \subfigure[\small{FreeBee}]
 {\label{fig:Transparancy_FreeBee}
 	\includegraphics[width=0.31\textwidth, height=0.15\textheight]{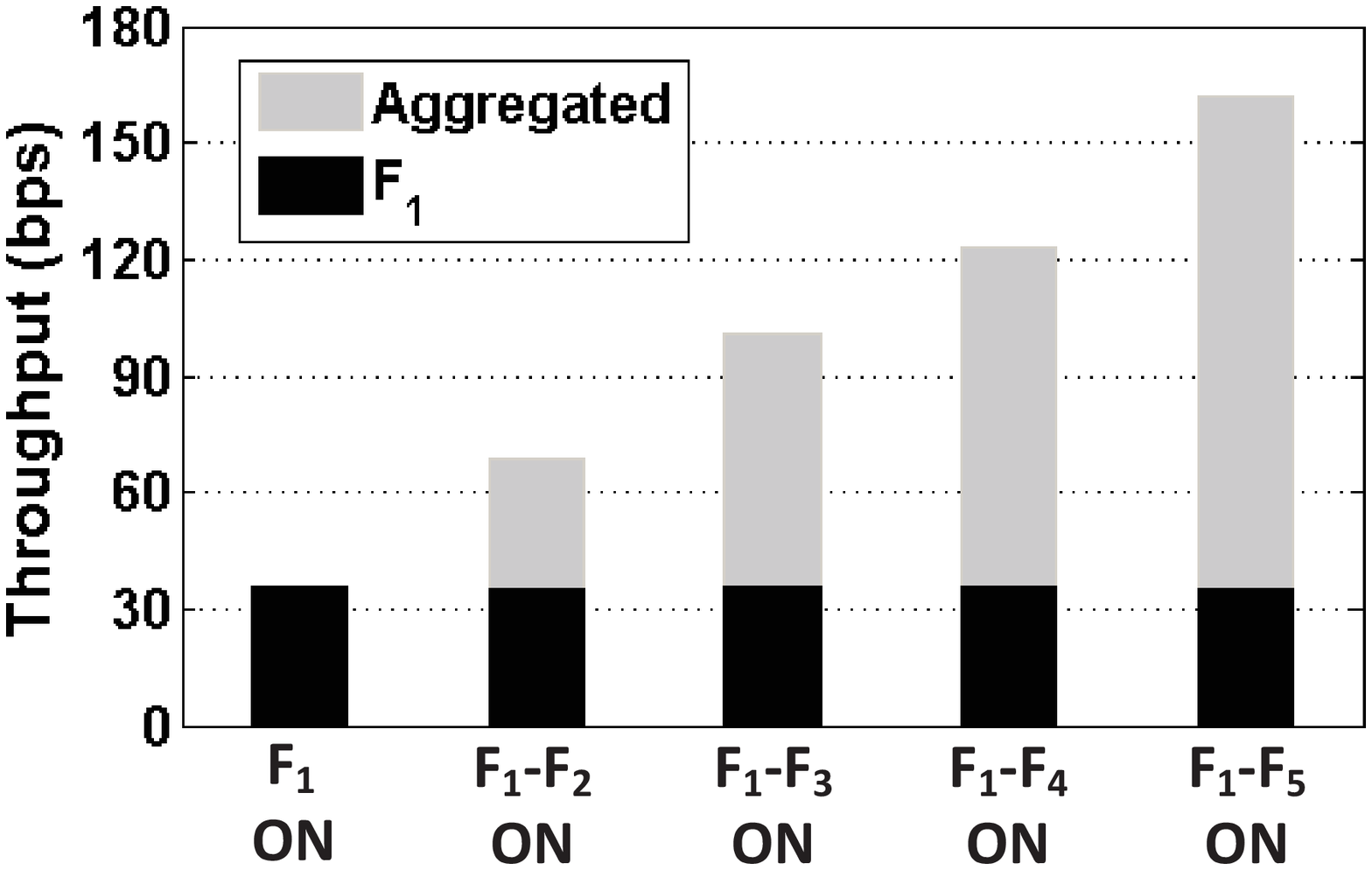}}
 	\hspace{2mm}
 \subfigure[\small{WiFi}]
 {\label{fig:Transparancy_WiFi}
		\includegraphics[width=0.31\textwidth, height=0.15\textheight]{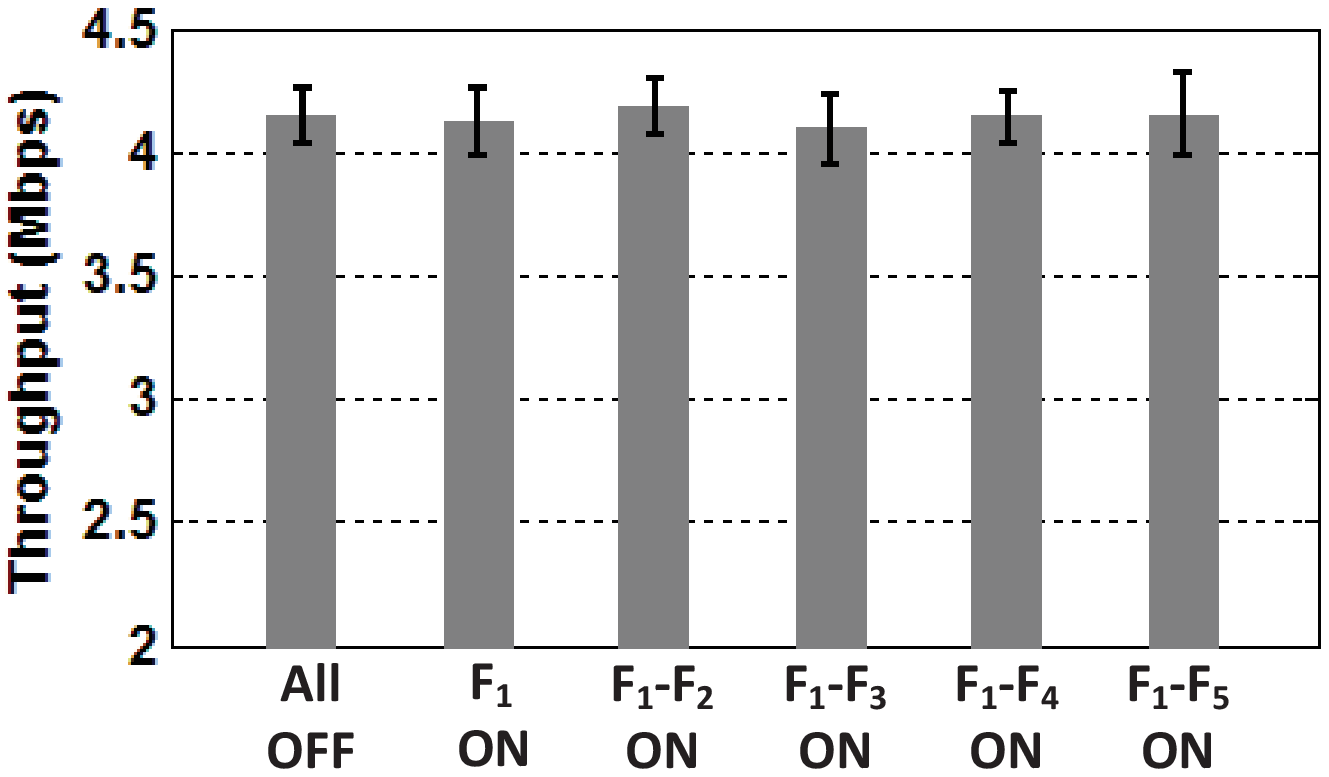}}
		\hspace{2mm}
 \subfigure[\small{ZigBee}]
 {\label{fig:Transparancy_ZigBee}
  \includegraphics[width=0.31\textwidth, height=0.15\textheight]{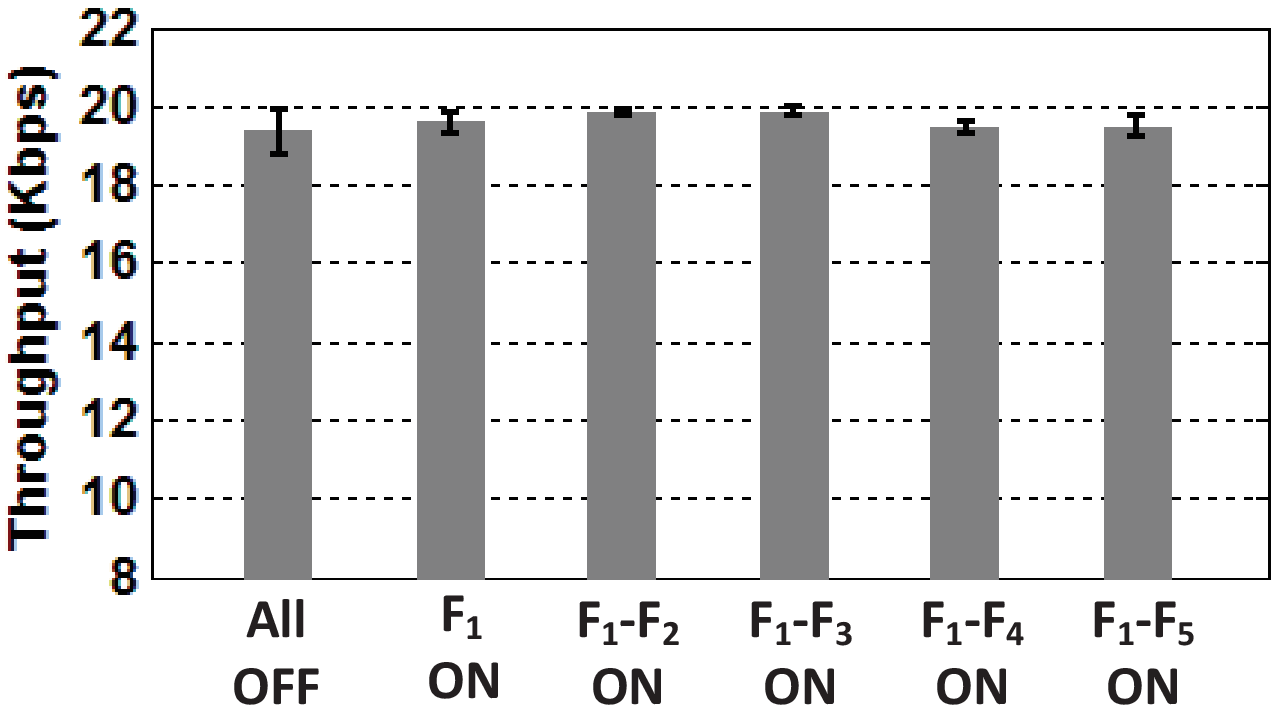}}		
	\vspace{-4mm}	
 \caption{Transparency of FreeBee to (a)concurrent FreeBee signals, (b)WiFi, and (c)ZigBee}
 \label{fig:Transparancy}
 \vspace{+2mm}
\end{minipage}
\end{figure*}



\noindent$\bullet$ \textbf{Cross-channel broadcast}: Figure~\ref{fig:Cross_Channel_Broadcast} demonstrates another FreeBee's feature where a sender with wider bandwidth (WiFi) reaches multiple narrower-band receivers (ZigBee) on different channels with a single broadcast. WiFi's bandwidth spans 20MHz while it is only 2MHz for ZigBee. The experiment result in Figure~\ref{fig:Cross_Channel_Broadcast} shows that SER on channel 14 is larger than that of channel 15. This is because channel 14 is affected by the noise from the heavily-used WiFi channel 1.

We note cross-technology broadcast and cross-channel bro-adcast can be combined to support more sophisticated scenarios where multiple heterogeneous senders deliver control symbols (through interval multiplexing) to multiple heterogeneous wireless receivers running under different channels. Such a capability would encourage further research on cross-technology coordination and control.


\begin{figure*}[t]
 \centering
\vspace{-0.1in}
\begin{minipage}[t][0.16\textheight]{1\textwidth}
 \centering
 \subfigure[\small{Walk}]
 {\label{fig:Mobile_Outdoor_Walk}
  \includegraphics[width=0.235\textwidth, height=0.14\textheight]{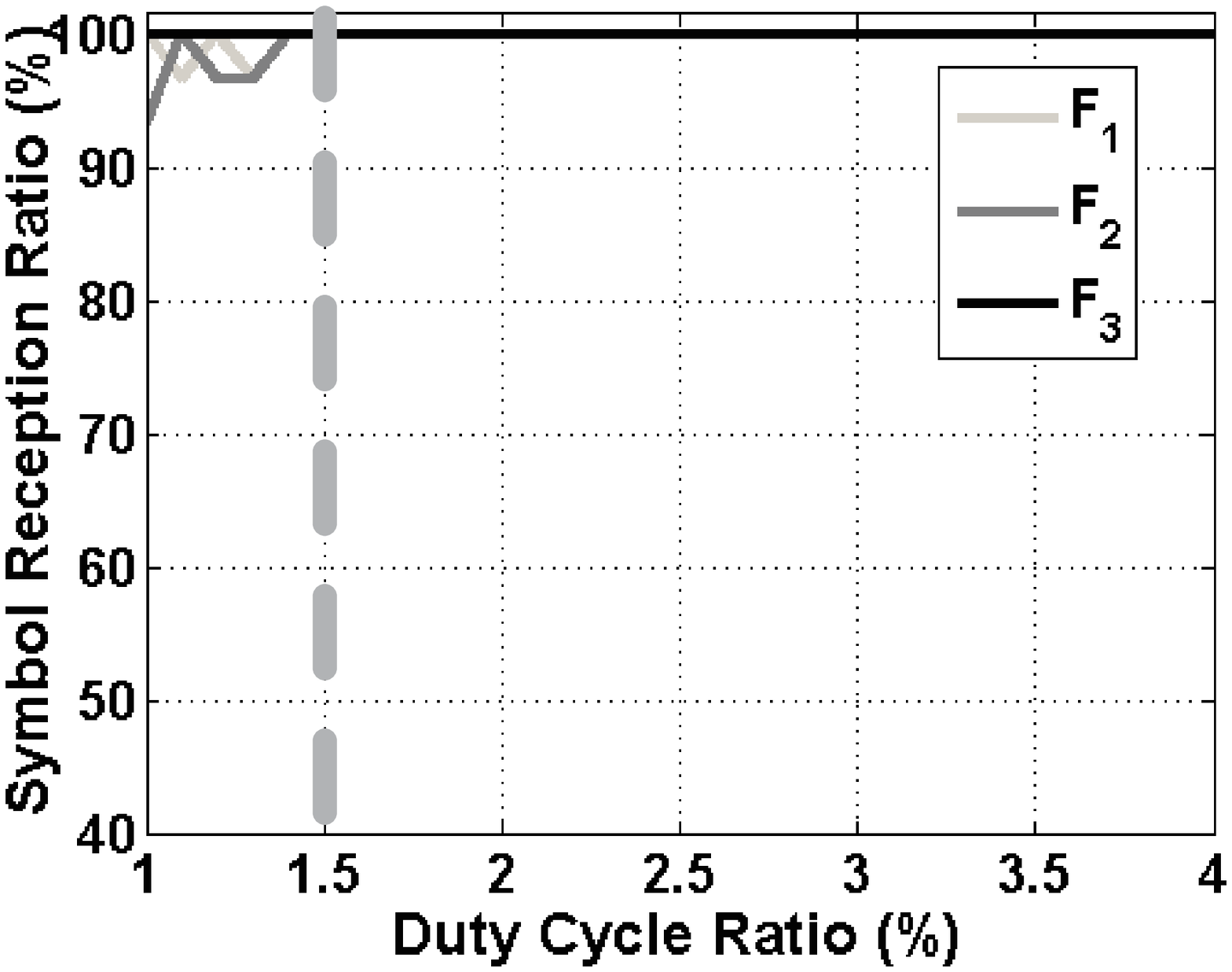}}
 \subfigure[\small{Run}]
 {\label{fig:Mobile_Outdoor_Run}
 	\includegraphics[width=0.235\textwidth, height=0.14\textheight]{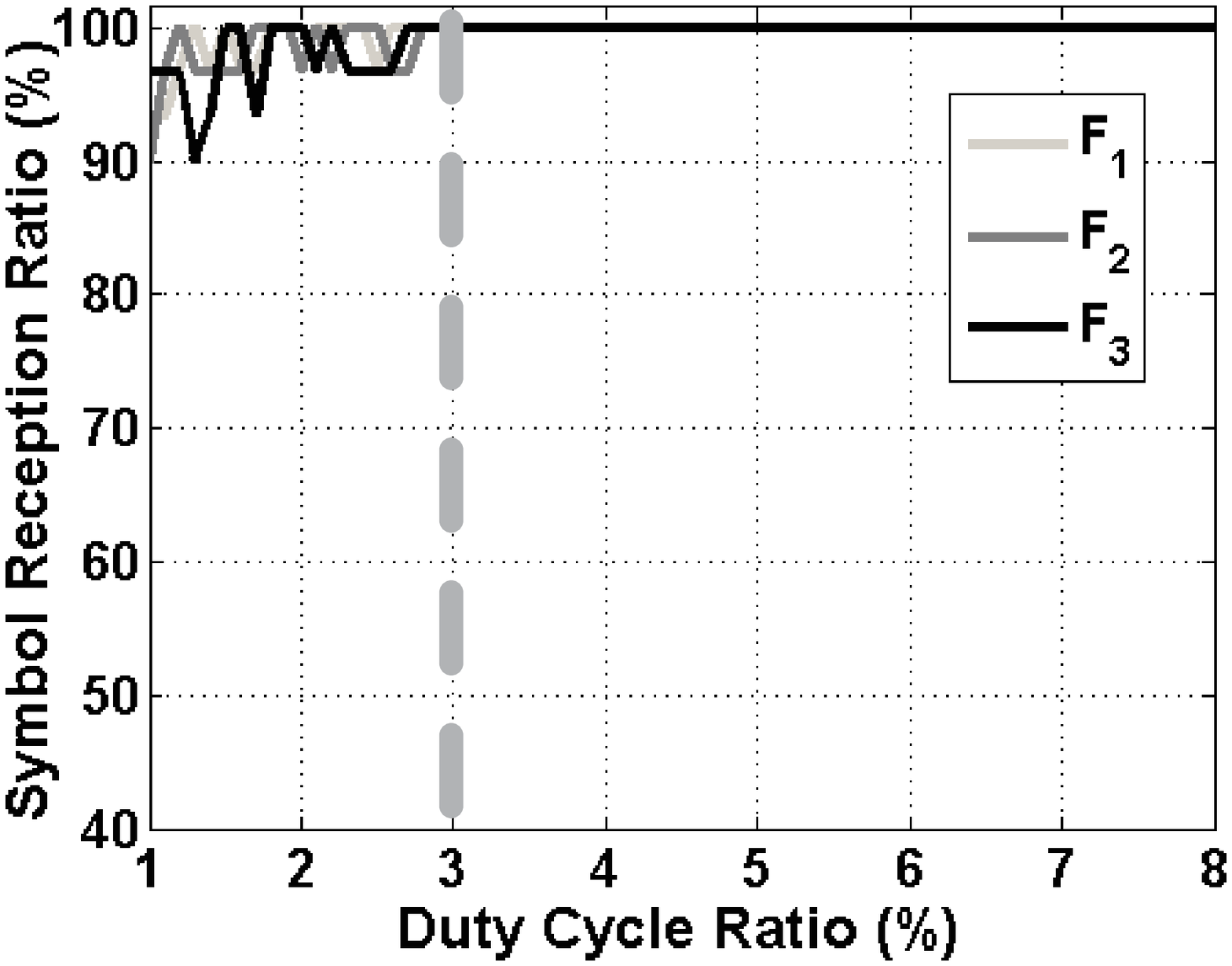}}
 \subfigure[\small{Bicycle}]
 {\label{fig:Mobile_Outdoor_Bike}
		\includegraphics[width=0.235\textwidth, height=0.14\textheight]{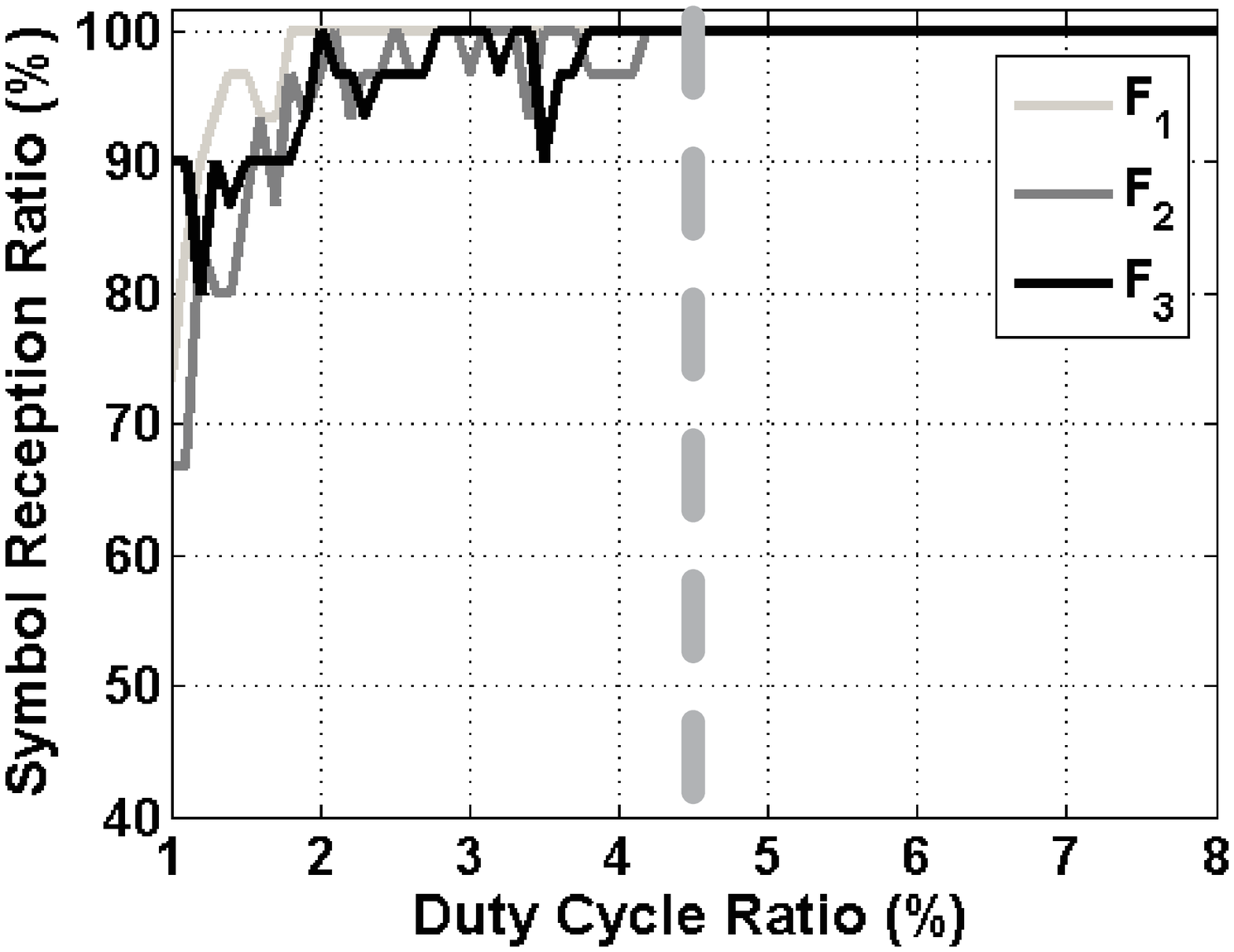}}
 \subfigure[\small{Car}]
 {\label{fig:Mobile_Outdoor_Car}
  \includegraphics[width=0.235\textwidth, height=0.14\textheight]{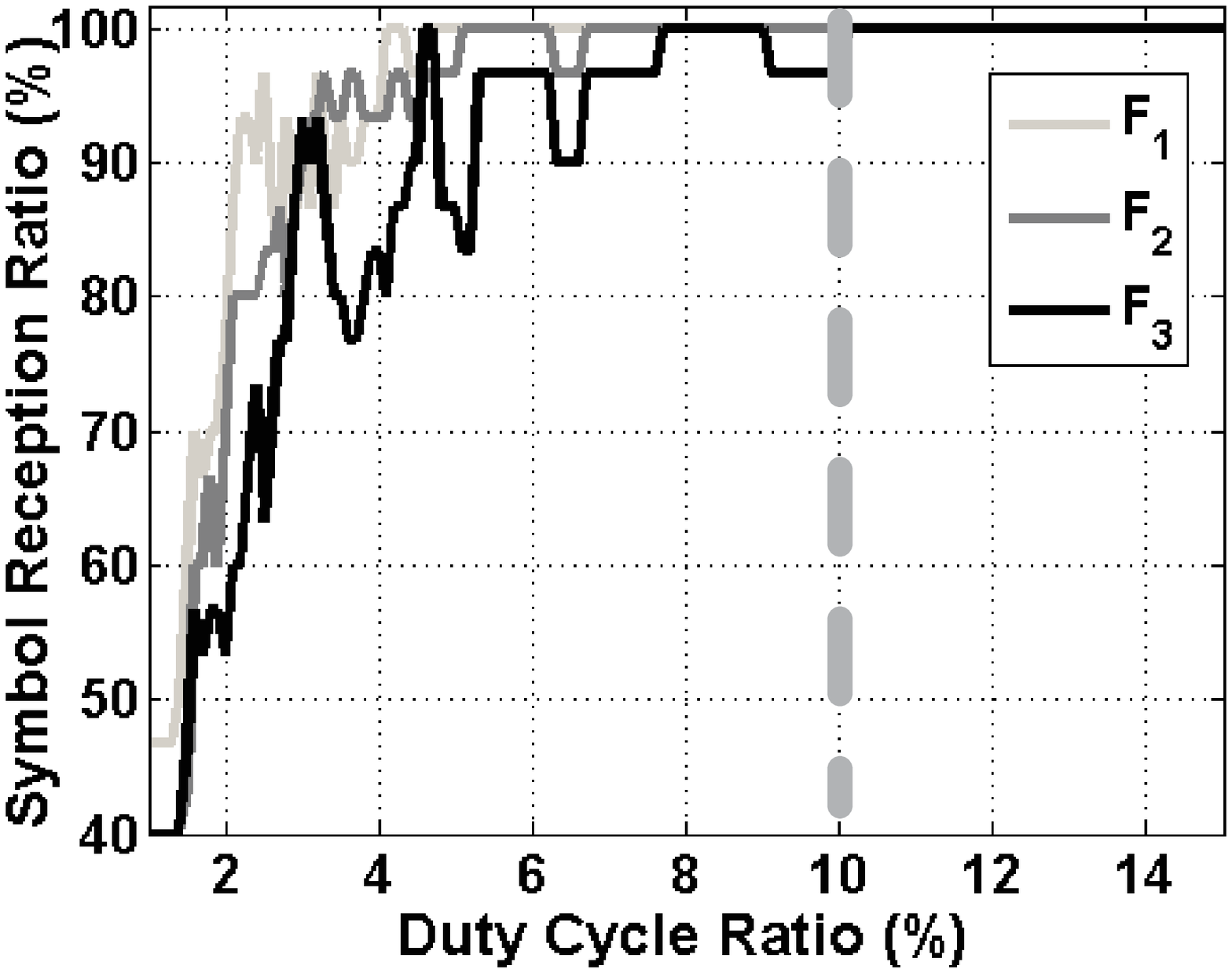}}		
\vspace{-0.1in}
 \caption{When walking, running, riding bicycle and car, reception can be ensured with duty cycling of only 1.5\%, 3\%, 4.5\% and 10\%, respectively.}
 \label{fig:Mobility_Result_Outdoor}
\end{minipage}
\end{figure*}

\subsection{Transparency}\label{sec:transparent}
We deploy five WiFi FreeBee senders, $F_1$-$F_5$, with beacon intervals of $\{91.1, 99.3, 103.4, 105.5, 109.6\}$ $ms$ corresponding to $\{89\Delta, 97\Delta, 101\Delta, 103\Delta, 107\Delta\}$. Prime intervals indicate they are interval-multiplexed. Each sender is allowed to operate in FreeBee ON or OFF modes. A sender embeds FreeBee symbols only when it is ON. When OFF, the sender acts as a legacy AP to simply transmit beacons periodically without symbol embedding. Under this setting, we observe the throughput of legacy wireless networks when a single to multiple FreeBee APs are operational, as well as the cross-interference among them. To do so, we first turn only $F_1$ to ON mode, while leaving all the others OFF. Then, we turn $F_2-F_5$ to ON one by one every 10 minutes.

\vspace{1mm}
\noindent$\bullet$ \textbf{Transparency between FreeBees}: As the black bar in Figure~\ref{fig:Transparancy_FreeBee} demonstrates, the throughput of $F_1$ is kept stable at an average of 35.5bps in the face of multiple concurrent FreeBee transmissions. This validates the performance of interval multiplexing in suppressing cross-interference among FreeBee signals. In fact simultaneous transmission helps to linearly increase the aggregated throughput up to 161.8bps, as shown by the gray bar. We note that this result is by selecting intervals close to that of a legacy AP (102.4$ms=100\Delta$) to limit the channel usage to a similar level.

The throughput is vastly enhanced by selecting shorter intervals for $F_1-F_5$ (i.e., by increasing the channel usage). We conduct experiments with various interval to validate the effects and limitations in practice. The result shows that the aggregated throughput for $F_1-F_5$ consistently increases with shorter intervals, where it reaches the maximum throughput of 536.6bps (i.e., more than 3x the standard interval case in Figure~\ref{fig:Transparancy_FreeBee}) with the intervals $\{5\Delta, 7\Delta, 11\Delta, 13\Delta, 17\Delta\}$. Further decreasing the interval, for example, to $\{2\Delta, 3\Delta, 5\Delta, 7\Delta, 11\Delta\}$, yields a lower throughput of 467.9bps due to the increased chance of cross-interference between FreeBee signals.

\vspace{1mm}
\noindent$\bullet$ \textbf{Transparency to legacy networks}: We again turn $F_1-F_5$ to ON mode one at a time every 10 minutes, during which throughput between a pair of WiFi or ZigBee nodes are measured. The results are demonstrated in Figures~\ref{fig:Transparancy_WiFi} and~\ref{fig:Transparancy_ZigBee}. For WiFi, we use \texttt{Iperf}~\cite{Iperf} to measure TCP throughput when operating in 802.11g. The WiFi sender was placed beside the FreeBees, 8m away from the WiFi receiver. For ZigBee measurement, sender and receiver were placed 5m apart, where 20byte packets were continuously transmitted at the fastest speed (i.e., with the minimum inter-packet delay). The figures show stable throughputs averaging 4.1Mbps and 19.6Kbps for WiFi and ZigBee, respectively, with small standard deviations. This suggests that both networks are unaffected by the presence of FreeBees, confirming that the free side-channel design successfully keeps FreeBee transparent to legacy networks.

\begin{figure}[h]
\centering
\begin{minipage}[t][0.37\textheight]{0.48\textwidth}
\centering
\subfigure[\small{Outdoor}]
{\label{fig:Mobile_Outdoor}
\includegraphics[width=1\textwidth]{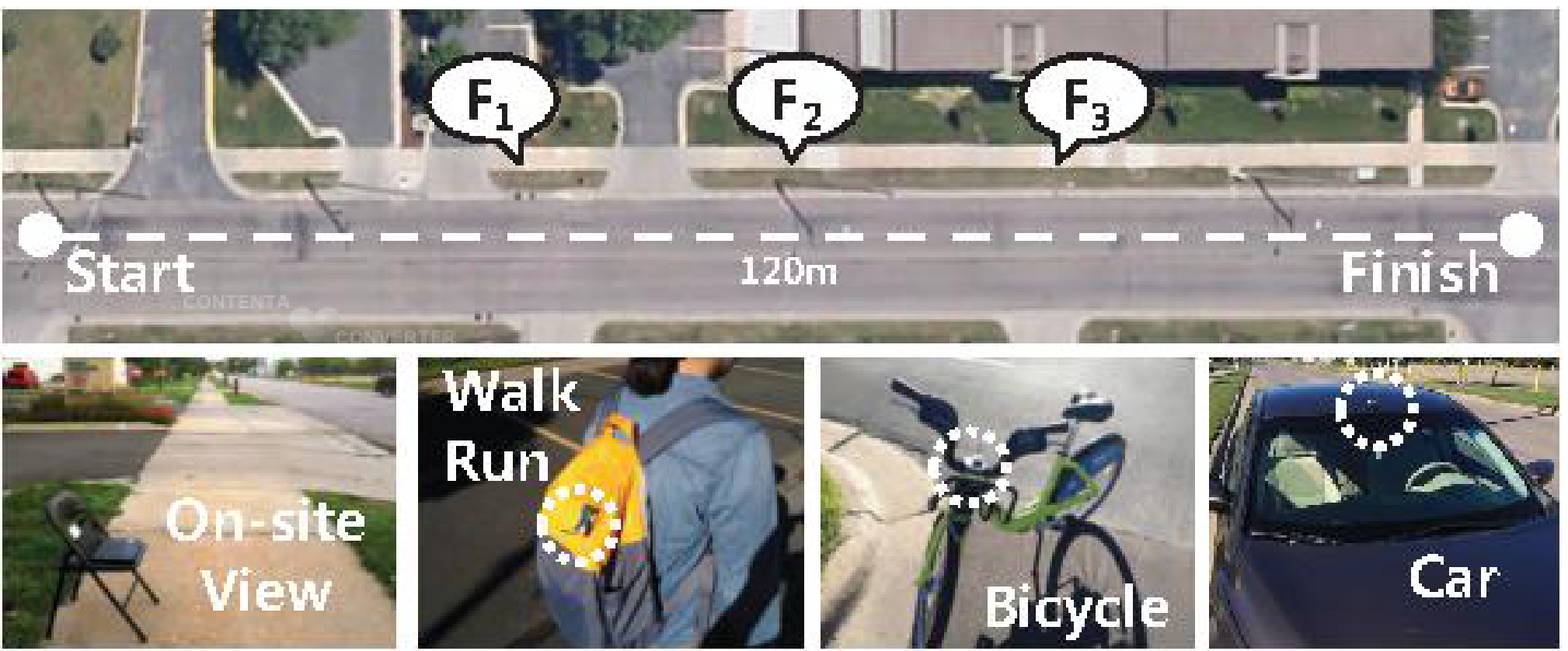}}
\subfigure[\small{Indoor}]
{\label{fig:Mobile_Indoor}
\includegraphics[width=1\textwidth]{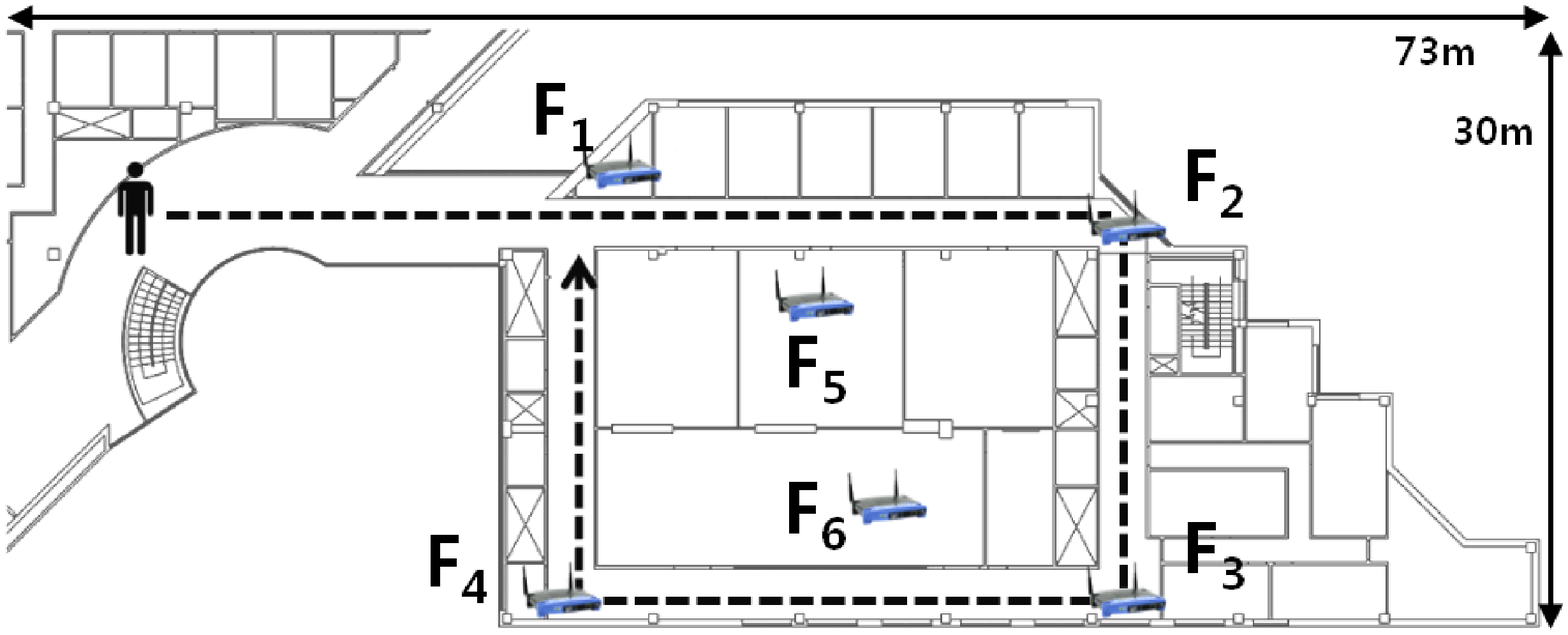}}
\caption{Mobile scenario evaluation: we experimentally reveal the impact of mobility on FreeBee on both indoor and outdoor scenarios, when the receivers are duty-cycled.}
\label{fig:Mobile_scenario_evaluation}
\end{minipage}
\end{figure}

\subsection{Mobility and Duty Cycling}\label{sec:mobile}

Here we consider a practical scenario of mobile receivers. For instance, this insight offers ideas on the feasibility of A-FreeBee for applications that requires location tracking. Furthermore, we take duty cycle into account where radios are turned off the majority of the time to preserve energy, a technique often adopted in battery-powered ZigBee networks to support long-term operations~\cite{cao20122, INFOCOM15WRx, LiLL14}. We deploy several WiFi A-FreeBee senders (laptop) in a university building and on a street, where the mobile receivers pass by the senders at different speeds: walking, running, on a bicycle, and in a car. It is worth noting that this experiment takes advantage of our free-side channel design, as (A-)FreeBee symbols can be continuously broadcast without occupying the channel, thus reaching mobile and/or duty-cycled receivers whose presence or active periods are unknown a priori.

\vspace{1mm}
\noindent$\bullet$ \textbf{On a street}: In this experiment, we deploy 3 A-FreeBee senders, 20m apart, on a street as shown in Figure~\ref{fig:Mobile_Outdoor}. The figure also shows ZigBee receivers with different degrees of mobility: walking (4.3 mph), running (6.8 mph), on a bicycle (10.8 mph) and in a car (30 mph). Thirty receivers were tested for each speed. Figure~\ref{fig:Mobility_Result_Outdoor} informs that A-FreeBee symbols from all the senders can be safely received with duty-cycling of only 1.5\%, 3\%, 4.5\% and 10\% when walking, running, cycling and riding in a car, respectively.


\begin{figure}[h]
  \centering
  \includegraphics[width=0.35\textwidth]{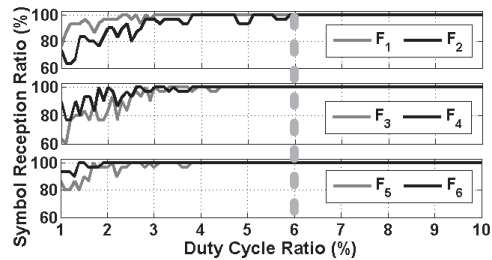}
  \caption{Duty cycling of 6\% guarantees reception in a university building.}
  \label{fig:Mobility_Result_Indoor}
\end{figure}

\noindent$\bullet$ \textbf{In a university building}: A total of 6 FreeBee senders are deployed as shown in Figure~\ref{fig:Mobile_Indoor} to show the impact of various positions, such as hallways and rooms with closed or open doors. A person walks at a moderate speed of 4.3 mph through the hallway following the arrow in the figure, with a ZigBee receiver tied to his/her backpack. The result after repeating for 30 different receivers is shown in Figure~\ref{fig:Mobility_Result_Indoor}. Overall, as indicated by the gray dotted line, duty cycling of 6\% ensures reception of A-FreeBee symbols from all the senders. A higher duty cycle ratio is required compared to the 1.5\% when walking outside in the previous experiment due to the noisy environment with 50+ legacy WiFi APs deployed.

\subsection{Receiver Overhead}\label{sec:receiverOverhead}
In this section we demonstrate light-weight storage and computation/energy overheads of (A-)FreeBee demodulation, making it affordable even to low-end devices (e.g., ZigBee mote) without disrupting the ongoing legacy protocols.

\noindent$\bullet$ \textbf{Storage}: Let us consider an example case where the sender is a WiFi AP with the beacon interval of $T=97\Delta=99.3ms$ (i.e., $\lambda$=776) and the receiver is a ZigBee mote. As explained earlier, the receiver samples RSSI every $128us$, where the total number of $\rho\times\lambda$ samples are collected for demodulation. As a RSSI value is represented as a single bit, under $\rho=5$ FreeBee requires $5\times776$=3.88Kbits (i.e., 485Bytes) of memory. This takes up less than 1\% of the storage offered in popular commodity ZigBee motes of MICAz and TelosB, where they offer 512KBytes and 1MBytes of on-board flash, respectively. Moreover, it is also important to note the followings: First, concurrent demodulation of $n$ interval-multiplexed FreeBee symbols whose intervals are $\lambda_1<\lambda_2<...<\lambda_n$, is achieved by collecting $\rho\times\lambda_n$ samples. Second, it is possible to reduce the memory usage by storing only the timing of transitions (i.e., 0$\leftrightarrow$1) in RSSI values. Third, A-FreeBee requires to store twice the number of samples compared to the basic FreeBee.

\noindent$\bullet$ \textbf{Computation/Energy}: The demodulation process can be decomposed into three parts: (i) sampling RSSI, (ii) adding to fold sum, and lastly, (iii) locating the position of the maximum fold sum. (i) and (ii) are repeatedly performed every $128us$, periodically incurring small overheads to perform a fetch and an addition operations. Upon collecting sufficient number (i.e., $\rho\times\lambda$) of samples (iii) takes place in which series of compare operations are executed in search for the maximum among $\lambda$ fold sums (for the basic FreeBee). Two largest values are sought among 2$\lambda$ for A-FreeBee, where we note that the time complexity remains linear to the number of fold sums in both versions.

We demonstrate the overhead of the entire steps (i.e., (i)-(iii)) on a off-the-shelf, low-end system -- ZigBee-compliant MICAz node. Specifically, we observe the execution time the node spends to process the demodulation. This is accurately measured by triggering a GPIO (general-purpose input/output) pin, whose activity is captured with Tektronix DPO 4054 oscilloscope. It is found that the demodulation of a symbol with an interval of $\lambda=776$ costs $51.5ms$ of computation time, consuming only $1.2mJ$ of additional energy compared to sleeping. The light overhead allows FreeBee to run on commodity low-end hardwares without disrupting the legacy protocols. Finally, we note that demodulating multiple, interval-multiplexed FreeBee symbols requires more computation/energy where they still remain affordable; For instance, listening to 10 concurrent symbols costs $98.5ms$ of computation, indicating $2.4mJ$ of energy.

\section{State of the Art} \label{sec:relatedworks}

This work lies in the intersection of (i) cross-technology interoperation and (ii) cross-technology communication, where we discuss the state-of-the-arts in the directions in relation to our work.

\subsection{Cross-technology Interoperation}
Improving network performance as well as bringing advanced services through cross-technology interoperation has been extensively investigated in many recent works~\cite{HaoZXM11, jin2011wizi, LiZH12, mahindra2014practical, sen2012dual, Sengul08, ZhouXXSM10}. They commonly adopt gateways, which are devices dedicated to bridging heterogeneous wireless technologies using their multiple radio interfaces. For example, WiBee~\cite{LiZH12} provides real-time coverage maps of multiple WiFi APs present. This is achieved by a ZigBee network that continuously monitors wireless signals in the air, and reports the observation to WiFi users through gateway. WiZi-Cloud~\cite{jin2011wizi} utilizes gateway to access the Internet through low-power ZigBee technology, which grants significant energy savings for mobile devices compared to when using WiFi. However, such gateway-based techniques commonly suffer from several inherent issues, including (i) extra cost to deploy the gateways as they are rarely available in practice, (ii) additional overhead due to in/outward traffic flowing into/out of the gateways, and (iii) deployment complexity related to the positioning of the gateways and their impact to the network performance.



\subsection{Cross-technology Communication}
Recently, a handful of work have proposed direct communication across wireless technologies, including GSense~\cite{ZhangS13}, Esense~\cite{ChebroluD09}, HoWiES~\cite{ZhangL13a}, C-Morse~\cite{Yin17}, and DCTC~\cite{Jiang17}. GSense enables cross-technology communication by prepending a customized preamble to legacy packets that contains a sequence of pulses in which gaps between them represent the data to be delivered. However, this design requires a special hardware, which hinders exploring the opportunity to utilize readily-deployed devices.

Esense and HoWiES, like our design, are compatible with off-the-shelf system; Esense establishes communication channels from WiFi to ZigBee by modulating packet lengths to those are unlikely to be used in usual WiFi traffic. HoWiES extends the Esense mechanism to convey data with combinations of length-modulated WiFi packets, and introduces a simple coding technique. HoWiES and Esense commonly require dummy packets to be pushed into the wireless channel to deliver cross-technology messages. However, due to their limited throughput of hundreds of bps, incurring traffic dedicated to the designs lead to a significant degradation of the spectrum efficiency. FreeBee does not suffer from this issue as it utilizes the mandatory beacons to enable free side-channel communication. C-Morse and DCTC also use existing traffic without introducing additional overhead, but FreeBee uniquely departs from them in a number of aspects such as (i) an exclusive feature of interval multiplexing that allows concurrent communication between multiple senders and receivers with implicit addressing, and (ii) providing the first running prototype showcasing a generic communication capability among three popular wireless technologies of WiFi, ZigBee, and Bluetooth.

We note that the FreeBee design of embedding information within the beacon timing is inspired by PPM, widely studied in the field of optical and UWB (Ultra Wide Band) communications ~\cite{carbonelli2006m, d2007energy, ghassemlooy1998digital, qiu2005ultra, shiu1999differential, wu2006weighted}. The extensive research in the area, however, is not applicable to our scenario because they are designed for pulses (which correspond to beacons in our case) that occur at their exact timing and/or can be differentiated from noise via a matched filter. Neither of these conditions is met by the beacons, since they are prone to channel access delays due to CSMA and they are indistinguishable from other traffic for receivers with incompatible PHY layers. 
 \section{Conclusion} \label{sec:conclusion}

A cross-technology communication framework of FreeBee is proposed, which aims to take advantage of the wireless coexistence via mutual supplementation. Extensive testbed experiments on three popular wireless technologies, WiFi, ZigBee, and Bluetooth, reveal that our design offers reliable symbol delivery within less than a second. Its unique free side-channel design and the utilization of short beacon frames grant high spectrum efficiency and feasibility; FreeBee exhibits 4.3$\times$ the throuhgput compared to the state-of-the-art, Esense, and demonstrates strong support for highly mobile (in a car) and extremely duty cycled ($\leq5\%$) receivers, implying its applicability to a wide range of practical applications. An examination of real WiFi deployment patterns in a shopping mall area finds that FreeBee can save 78.9\% of the energy otherwise wasted by the WiFi interface.

\ifCLASSOPTIONcaptionsoff
  \newpage
\fi

\bibliographystyle{abbrv}
\bibliography{survey}

\end{document}